\documentclass[11pt]{article}
\usepackage{fullpage,times}
\usepackage{amsmath,amsfonts,amssymb,amsthm}
\usepackage{hyperref}
\usepackage[capitalise]{cleveref}
\usepackage{kpfonts}
\usepackage{mathtools}
\usepackage[utf8]{inputenc}
\usepackage{nicefrac}
\usepackage{xcolor}
\usepackage{bbm}
\usepackage{enumitem}
\usepackage{graphicx}
\usepackage{wrapfig}
\usepackage[normalem]{ulem}
\newcommand{\calA}{{\cal A}}
\newcommand{\calD}{{\cal D}}
\renewcommand{\epsilon}{\varepsilon}
\DeclareMathOperator{\convert}{convert}
\DeclarePairedDelimiter{\myangle}{\langle}{\rangle} 
\DeclarePairedDelimiter{\setsize}{\mid}{\mid}
\DeclarePairedDelimiter{\ceil}{\lceil}{\rceil}
\newcommand{\bit}{\set{0,1}}
\newcommand{\remove}[1]{}
\renewcommand{\cref}{\Cref}
\DeclareSymbolFont{AMSb}{U}{msb}{m}{n}
\DeclareMathSymbol{\N}{\mathbin}{AMSb}{"4E}
\DeclareMathSymbol{\Z}{\mathbin}{AMSb}{"5A}
\DeclareMathSymbol{\R}{\mathbin}{AMSb}{"52}
\DeclareMathSymbol{\Q}{\mathbin}{AMSb}{"51}
\DeclareMathSymbol{\erert}{\mathbin}{AMSb}{"50}
\DeclareMathSymbol{\I}{\mathbin}{AMSb}{"49}
\DeclareMathSymbol{\C}{\mathbin}{AMSb}{"43}

\newcommand{\paragraphnegative}{\paragraph}

\newcommand{\AAA}{\mathcal A}
\newcommand{\BBB}{\mathcal B}

\newcommand{\DDD}{\mathcal D}

\newcommand{\eps}{\varepsilon}

\newcommand{\error}{{\rm error}}
\newcommand{\db}{S}
\newcommand{\supp}{\operatorname{\rm support}}

\newcommand{\VC}{\operatorname{\rm VC}}

\newcommand{\poly}{\mathop{\rm poly}}

\newcommand{\set}[1]{\left\{ #1 \right\}}

\def\1{\mathbbm{1}}

\def\Q{\operatorname*{\mathbb{Q}}}
\def\poly{\mathop{\rm{poly}}\nolimits}

\newcommand{\inr}{\in_{\mbox{\tiny R}}}

\newtheorem{theorem}{Theorem}[section]
\newtheorem*{theorem*}{Theorem}
\newtheorem{lemma}[theorem]{Lemma}
\newtheorem*{lemma*}{Lemma}
\newtheorem{claim}[theorem]{Claim}

\newtheorem{definition}[theorem]{Definition}

\theoremstyle{definition}

\newtheorem{fact}[theorem]{Fact}
\newtheorem{remark}[theorem]{Remark}

\newcommand{\kn}[1]{}
\newcommand{\us}[1]{}
\newcommand{\ak}[1]{}

\newcommand{\Threshold}{{\sf Threshold}}
\newcommand{\Parity}{{\sf Parity}}
\newcommand{\Thr}{\mathsf{Thr}}
\newcommand{\Par}{\mathsf{Par}}
\newcommand{\msg}{M}

\crefname{proposition}{Proposition}{Propositions}
\crefname{claim}{Claim}{Claims}
\crefname{remark}{Remark}{Remarks}
\newcommand{\stepref}[1]{Step~\ref{#1}}

\title{The power of synergy in differential privacy: \\ Combining a small curator with local randomizers}

\author{
Amos Beimel\thanks{Dept.\ of Computer Science, Ben-Gurion University. {\tt amos.beimel@gmail.com}}
\and
Aleksandra Korolova\thanks{
Dept.\ of Computer Science,
University of Southern California.
\tt{korolova@usc.edu}
}
\and
Kobbi Nissim\thanks{Dept. of Computer Science, Georgetown University. \tt{kobbi.nissim@georgetown.edu}}
\and
Or Sheffet\thanks{ 
Faculty of Engineering, 
Bar-Ilan University.
\tt{or.sheffet@biu.ac.il}
}
\and
Uri Stemmer\thanks{Dept. of Computer Science, Ben-Gurion University and Google Research. {\tt u@uri.co.il}}
}

\begin{document}

\thispagestyle{empty}
\maketitle

\begin{abstract}

Motivated by the desire to bridge the utility gap between local and trusted curator models of differential privacy for practical applications, we initiate the theoretical study of a hybrid model introduced by ``Blender'' [Avent et al.,\ USENIX Security '17], in which differentially private protocols of $n$ agents that work in the local-model are assisted by a differentially private curator that has access to the data of $m$ additional users. We focus on the regime where $m\ll n$ and study the new capabilities of this \emph{$(m,n)$-hybrid} model. We show that, despite the fact that the hybrid model adds no significant new capabilities for the basic task of simple hypothesis-testing, there are many other tasks (under a wide range of parameters) that can be solved in the hybrid model yet cannot be solved either by the curator or by the local-users separately. Moreover, we exhibit additional tasks where at least one round of \emph{interaction} between the curator and the local-users is necessary -- namely, no hybrid model protocol without such interaction can solve these tasks. Taken together, our results show that the combination of the local model with a small curator can become part of a promising toolkit for designing and implementing differential privacy.
\end{abstract}

\thispagestyle{empty}
\newpage

\setcounter{page}{1}

\section{Introduction}\label{sec:intro}

Data has become one of the main drivers of innovation in applications as varied as technology, medicine~\cite{ginsberg2009detecting}, and city planning~\cite{bettercities, nycbettercities}. However, the collection and storage of personal data in the service of innovation by companies, researchers, and governments poses significant risks for privacy and personal freedom.
Personal data collected by companies can be breached~\cite{BreachStats}; 
subpoenaed by law enforcement in broad requests~\cite{ReverseLocationWarrants2019}; mis-used by companies' employees~\cite{UberGodView, CapitalOneHacker}; or used for purposes different from those announced at collection time~\cite{venkatadri-2019-pii}. 
These concerns alongside data-hungry company and government practices have propelled privacy to the frontier of individuals' concerns, societal and policy debates, and academic research.

The \textit{local model of differential privacy}~\cite{Warner65, KLNRS11} has recently emerged as one of the promising approaches for achieving the goals of enabling data-driven innovation while preserving a rigorous notion of privacy for individuals that also addresses the above challenges. 
The {\em differential privacy} aspect provides each participating individual (almost) with the same protection she would have had her information not been included~\cite{DMNS06}, a guarantee that holds even with respect to all powerful adversaries with access to multiple analyses and rich auxiliary information.
The {\em local} aspect of the model means that this guarantee will continue to hold even if the curator's data store is fully breached. From the utility perspective, the deployment of local differential privacy protocols by Google, Apple, and Microsoft demonstrate that the local differential privacy model is a viable approach 
in certain applications, without requiring trust in the companies or incurring risks from hackers and intelligence agencies~\cite{RAPPOR, cnet, wired, Apple, ding2017}. 

The adoption of the local model by major industry players has motivated a line of research in the theory of local differential privacy (e.g.,~\cite{BassilyS15,BNST17,BNS17,Stemmer19,JosephMNR19,JosephRUW18,JosephMR19}).
Alongside algorithmic improvements, this body of work highlighted the wide theoretical and practical gap between utility achievable in the more traditional trusted curator model of differential privacy (where the curator ensures the privacy of its output but can perform computations on raw individual data) and that achievable in the local differential privacy model. In particular, the number of data points necessary to achieve a particular level of accuracy in the local model is significantly larger than what is sufficient for the same accuracy in the curator model (see, e.g.,~\cite{BeimelNO08, KLNRS11, ShiCRS17, BassilyS15}). This has negative consequences. First, data analysis with local differential privacy becomes the privilege of the data-rich, handicapping smaller companies and helping to cement monopolies. Second, in an effort to maintain accuracy the entities deploying local differential privacy are tempted to use large privacy loss parameters~\cite{appleepsilon}, ultimately putting into question the privacy guarantee~\cite{wired2017}.

New models for differentially private computation have recently emerged to alleviate the (inevitable) low accuracy of the local model, of which we will discuss the \textit{shuffle model}~\cite{IKOS06, prochlo, CheuSUZZ19, BalleBGN19, BalleBGNnote, GhaziPV19} and the \textit{hybrid model}~\cite{blender}.\footnote{Approaches which weaken differential privacy itself or justify the use of large privacy loss parameters are outside our scope and deserve a separate discussion. 
}
In the shuffle model, it is assumed that the curator receives data in a manner disassociated from a user identifier (e.g., after the raw data has been stripped of identifiers and randomly shuffled).  
Recent work has proved that protocols employing shuffling can provide better accuracy than local protocols and sometimes match the accuracy of the curator model~\cite{CheuSUZZ19, BalleBGN19, BalleBGNnote, GhaziPV19}.\footnote{Furthermore, the shuffle model provides new ``privacy amplification" tools that can be used in the design of differentially private algorithms~\cite{amplificationbyshuffling, BalleBGN19}.}
Although the shuffle model is a promising approach for bridging the gap between the local and trusted curator model, it suffers from two weaknesses: it requires individuals' trust in the shuffler (which itself may be subject to breaches, subpoenaes, etc.\ and the infrastructure for which may not exist), and, as highlighted by~\cite{BalleBGN19}, its privacy guarantees to an individual rely on the assumption that sufficiently many other individuals do not deviate from the protocol.

The focus of this work is on a generalization of the hybrid model introduced by~\cite{blender}, where a majority of individuals that participate in a local differential privacy computation is augmented with a small number of individuals who contribute their information via a trusted curator. 
From a practical point of view, this separation is aligned with current industry practices, and the small number of individuals willing to contribute via a curator can be employees, technology enthusiasts, or individuals recruited as alpha- or beta-testers of products in exchange for early access to its features or decreased cost~\cite{microsoft-opt-in, WindowsOptin, FirefoxOptin}. 

\begin{wrapfigure}{r}{0.45\textwidth} 
    \centering
     \vspace{-1.7cm} \includegraphics[width=0.46\textwidth]{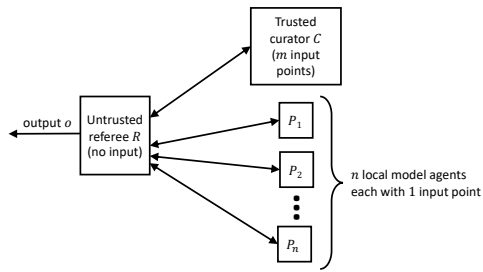} \vspace{-2.cm}
   \caption{The hybrid model\label{fig:model}.}
   \vspace{-0.5cm}
\end{wrapfigure}

Furthermore, unlike Blender~\cite{blender}, in an effort to explore a rich trust and communication model, and anticipate development of future technologies and practices, we do not assume that the curator trusted by the opt-in users and the eventual data recipient (whom we call the referee) are the same entity (see \cref{fig:model}). The detailed discussion of the benefits of this modeling assumption appears after the formal model definition in \cref{sec:hybrid-defn}.

\subsection{This work: The power of the hybrid model} 

We initiate a theoretical study of the extent to which the hybrid model improves on the utility of the local model by addition of a small curator (and vice versa, improves on the utility of a small curator by addition of parties that participate in a local model computation). We ask whether there are tasks that cannot be computed privately by a small curator, or in the local model (even with many participants), but are feasible in a model that
combines both. We answer this question in the affirmative.

\paragraph{A concatenation of problems (direct-sum).}
Our first example is composed of two independent learning tasks, each task can be  privately learned in one model, however, cannot be privately learned in the other model. Each data point is parsed as two points $(x,y)$ that are labeled $(\Par_k(x), \Thr_t(y))$ where the former is a parity function $\Par_k(x)=\langle k,x\rangle$ and the latter is a threshold function $\Thr_t(y) = \1_{\{y\geq t\}}$ (see \cref{sec:prelim-lt-pl} for complete definitions). For a choice of parameters, known sample complexity bounds imply that the parity learning part of the task can be performed by the curator but cannot be privately computed in the local model with sub-exponential number of messages~\cite{KLNRS11}, and, conversely, that the threshold learning part cannot be performed by the curator~\cite{BNS13,FeldmanX14} but can be performed by the local model parties. It follows that for our choice of parameters the combined task cannot be performed  neither by the curator nor by the local model with sub-exponential number of rounds 
(as the number of local agents is small, each agent in the local model must send many messages), 
but is feasible in the hybrid model without interaction (see \cref{sec:parity_concatenated_threshold} for a detailed analysis). 

\paragraph{Select-then-estimate.} 

In our second example, each input point $x$ is sampled i.i.d.\ from an unknown distribution over $\set{-1,1}^d$. Letting $\mu = E\left[x\right]$, the goal is to output a pair $(i,\hat\mu_i)$ where $i$ approximately maximizes $\mu_i$ (i.e., $\mu_i \geq \max_j \mu_j - \alpha$) and $\hat \mu_i$ approximates $\mu_i$ (i.e., $|\hat\mu_i - \mu_i|\leq \alpha'$), where $\alpha' < \alpha$. That is, once we found a good coordinate $i$,
we want to approximate its quality with smaller error.
A number of sample complexity bounds apply here (in particular,~\cite{Ullman18} and \cref{sec:selectionLowerbound}), yielding a wide range of parameters where the task cannot be performed by the curator alone, nor by the local model parties alone. The hybrid model, however, allows for a solution where the curator first identifies $i$ such that $\mu_i \geq \max_j \mu_j - \alpha$ and the estimation of $\mu_i$ within accuracy $\alpha'$ is then performed by the local model parties (see \cref{sec:selectThenEstimate} for a detailed analysis).

The select-then-estimate problem is a sub-component of many statistical and data-mining applications~\cite{bafna2017price, steinke2017tight, blender}. It is also of core interest in genome-wide association studies (GWAS), where the inputs are genomes of patients affected by a particular disease and the goal is to (approximately) identify disease-related genome locations and estimate their significance~\cite{azencott2018machine, yu2014scalable, bonte2018towards}.
Solving the problem while ensuring privacy is particularly challenging when the feature dimension is large compared to the number of available samples, which is the case for GWAS~\cite{berger2019emerging, azencott2018machine, johnson2013privacy, simmons2016realizing}.
As genomic data is particularly sensitive, the hybrid model of differential privacy appears appropriate for it from the privacy perspective -- the majority of the data would be analyzed with the guarantee of local differential privacy, and only a small number of data points would be entrusted to a curator, whose analysis's output   should also be differentially private~\cite{wang2009learning}.
As our example suggests, the hybrid model may be useful also from the utility perspective.

\bigskip

\noindent We next present and study tasks that 
require protocols with  different interaction patterns involving both the curator and the local agents;  that is, the 
referee needs to relay  messages from the curator to the local model agents in one problem or vice versa in a second problem.

\paragraph{Learning parity XOR threshold.} This task, which is a twist on the above concatenation of problems, 
 combines two independent learning tasks. 
 Rather than merely concatenating, in the learning parity XOR threshold problem
 points are labeled by  $\Par_k(x)\oplus\Thr_t(y)$. A simple argument shows that (for specifically chosen parameters) the task cannot be performed by the curator alone or by the local model agents with sub-exponential number of rounds. The task is, however, feasible in the hybrid model without interaction. Observe that the local model agents can complete the task once the parity part is done, and that $\Thr_t(y)=0$ for the lower half of the points $y$ or $\Thr_t(y)=1$ for the upper half of the points $y$ (or both). These observations suggest a protocol where the curator first performs two parity learning tasks (splitting points according to $y$ values), and the task is then completed by the local model agents. This requires communication (as the referee needs to relay a message from the curator to the local model agents), and it may seem that this interaction is necessary for the task. However, in \cref{sec:learnParityXorThreshold} we show that this is not the case, by demonstrating a non-interactive protocol where all parties send a message to the referee at the same round. 

\paragraph{1-out-of-$2^d$-parity.} 

Our next task can be solved in the hybrid model (but neither by a small curator model nor in the local model with sub-exponential number of rounds). This task requires interaction, first with the local model agents and then with the curator. The task consists of a multitude of parity problems, only one of which -- determined by the input -- needs to be solved. 
The curator is capable of solving the parity task privately, however, the curator's limited number of samples does not allow solving all parity problems privately, nor does it allow distinguishing which problem to solve. The local model agents cannot solve parity (with sub-exponential number of rounds)~\cite{KLNRS11} but can recover which problem needs to be solved (via a heavy hitters computation~\cite{BNS17}). Thus, the referee needs to first interact with the local agents and then the curator. See \cref{sec:1-out-of-2d-parity}.

\paragraph{Parity-chooses-secret.} 

The third task in this part 
can be solved in the hybrid model (but neither with a small curator, nor in the local model with sub-exponential number of rounds). The task requires interaction in the reverse order from the previous task: first with the curator, then with the local model agents. The input to this task contains shares of a large collection of secrets and the goal is to recover one of these secrets. The secret that should be recovered is determined by the input as the solution to a parity problem.
The curator can solve the parity problem but does not have enough information to recover any of the secrets. The local model agents receive enough information to recover all secrets, but doing so would breach privacy (as implied by~\cite{MMPRTV10}).
They cannot solve the parity problem with sub-exponential number of rounds. In the hybrid model protocol, the curator first solves the parity problem privately, and relates the identity of the secret to be recovered through the referee to the local model agents who then can send enough information to recover the required secret. See \cref{sec:paritySecretSharing}.

\medskip

The latter two tasks highlight information and private-computation gaps between the curator and the local model agents. The local model agents receive enough information to solve the task, but lack the ability to privately solve an essential sub-task (when they are not allowed to use exponentially many rounds). The curator does not receive enough information to solve the task (even non-privately), but can solve the hard sub-task.

\paragraph{When the hybrid model does not help much.}

Although most of the results in this work are on the positive side, demonstrating that cleverly utilizing a curator in synergy with local agents can allow for new capabilities, we also show that for one of the most basic tasks -- namely, basic hypothesis testing -- the hybrid model has no significant advantage over what can be done separately in either the local model with $m$ agents or in the curator model with database of size $n$. We show that if for two distributions $\calD_0$ and $\calD_1$ there is a private protocol in the hybrid model that given i.i.d.\ samples from $\calD_j$ correctly identifies $j$, then  there is a private protocol with the same sample size either in the local model or in the curator model that correctly identifies $j$ (with some small loss in the success probability).
We then consider two distributions $\calD_0$ and $\calD_1$  over the the domain $\set{0,1}$, where in $\calD_0$ the probability of $1$ is strictly less than $1/2$ and in $\calD_1$ the probability of $1$ is strictly greater than $1/2$ and  identify values of $m$ and $n$ such that in the hybrid model, where the curator has $m$ samples and there are $n$ agents (each holding one sample), no protocol exists that can differentiate whether the $m+n$ inputs were sampled i.i.d.\ from $\calD_0$ or from $\calD_1$.
Since computing the sum of i.i.d.\ sampled bits from $\calD_0$ or $\calD_1$ can distinguish between these distributions, the above results imply that for computing the sum, the hybrid model is no better than each model separately.
See \cref{sec:basic_hypothesis_testing}.

\paragraph{A new lower bound for selection in the local model.}

As mentioned above, our analysis for the select-then-estimate task relies on lower bounds on the sample complexity of selection in the local model. Ullman \cite{Ullman18} gives a (tight) lower bound of $\Omega(\frac{d \log d}{\alpha^2 \epsilon^2})$ samples for the non-interactive case.
In \cref{sec:selectionLowerbound} we show that for {\em interactive} local model protocols, the number of messages in such protocol is $\Omega(d^{1/3})$. For example, if the curator interacts with the local model parties so that each party sends $t$ messages, then the number of parties must be at least $\Omega(d^{1/3}/t)$. The proof is by a generic reduction from any private PAC learning task to selection, which preserves sample complexity. The bound is obtained by applying the reduction from privately learning parity and applying bounds on the sample complexity of privately learning parity from~\cite{KLNRS11}.

\subsection{Discussion and future work}

Our results show that the combination of the local model with a small curator can become part of a promising toolkit for designing and implementing differential privacy. More work is needed to develop the theory of this model (and possibly introduce variants), and, in particular, characterize which tasks can benefit from it. From an algorithms design perspective, now that we know that the hybrid model can lead to significant improvements over both the curator and local models, an exciting open question is  understanding what other non-trivial algorithms can be designed that take advantage of the synergy.

\paragraph{Selection bias.}
In this work we assume that the inputs for the curator and for the local model agents come from the same distribution. However, 
the recruitment of individuals for participating via a curator can create unintended differences between the input distributions seen by the curator and the entire population, and hence lead to biases, an issue which is outside the scope of the current work.
Selection bias remains an important issue that needs to be addressed.

\paragraph{Approximate differential privacy.}
Our separations between the hybrid model and the curator and local models hold for pure differential privacy (i.e., $\epsilon$-differential privacy). Specifically, we use lower bounds for $\epsilon$-differential private learning of the threshold functions in the curator model~\cite{BNS13, FeldmanX14}; these lower bounds do not hold  for $(\epsilon,\delta)$-differential private learning of the threshold functions~\cite{BNS16a,BNSV15}. We also use lower bounds for $\epsilon$-differential private learning of  parity in the local model~\cite{KLNRS11}; it is open if these lower bounds hold for fully interactive $(\epsilon,\delta)$-differential private learning protocols of parity. Possible separations for  $(\epsilon,\delta)$-differential privacy are left for future research.

\section{Preliminaries}
\subsection{Protocols for differentially private computations}
Let $X$ be a data domain. We consider a model where the inputs and the computation are distributed among parties $P_1,\ldots,P_n$. Each party is an interactive randomized functionality: it can receive messages from the other parties, perform a randomized computation, and send messages to the other parties. At the beginning of a computation, each party $P_i$ receives its input $\mathbf{x}_i = (x_{i,1},\ldots,x_{i,\ell_i}) \in X^{\ell_i}$. I.e., the input of party $P_i$ consists of a sequence of $\ell_i\geq 0$ entries taken from the data domain $X$, and the entire joint input to the protocol is $(x_{1,1}\ldots,x_{1,\ell_1},x_{2,1},$ $\ldots,x_{2,\ell_2},\ldots,x_{n,1},\ldots,x_{n,\ell_n})$. 
The parties  engage in a randomized interactive protocol $\Pi=(\Pi_{P_1},\ldots,\Pi_{P_n})$, where a message sent by a party $P_i$ in some round is  computed according to $\Pi_{P_i}$ and depends on its input $\mathbf{x}_i$, its random coins, and the sequence of messages it has seen in previous rounds. When a party $P_i$ halts, it writes its output to a local output register $\mathbf{o}_{P_i}$.
The number of messages in a protocol is the number of rounds multiplied by the number of parties. 
\begin{definition}
We say that $\mathbf{x}=(x_1,\ldots,x_\ell)\in X^\ell$ and $\mathbf{x}'=(x'_1,\ldots,x'_\ell)\in X^\ell$  are {\em neighboring} if they differ on at most one entry, i.e., there exist $i^*\in[\ell]$ such that $x_i = x'_i$ for $i\in [\ell]\setminus\set{i^*}$. 
\end{definition}

\begin{definition}
We say that two probability distributions ${\cal D}_0, {\cal D}_1\in \Delta(\Omega)$ are $(\epsilon,\delta)$-close and write ${\cal D}_0 \approx_{\epsilon,\delta} {\cal D}_1$ if $$\Pr_{t\sim {\cal D}_b}\left[t\in T\right] \leq e^{\epsilon} \cdot \Pr_{t\sim {\cal D}_{1-b}}\left[t\in T\right] + \delta$$ for all measurable events $T\subset \Omega$ and $b\in\set{0,1}$. 
\end{definition}

We are now ready to define what it means for a protocol to be differentially private in a fully malicious setting, i.e., in a setting where an arbitrary adversary controls the behavior of all but one party. Intuitively, a protocol is differentially private in a fully malicious setting if there do not exist a party $P_i$ and an adversary $A$ controlling $P_1,\ldots,P_{i-1},P_{i+1},\ldots,P_n$ such that $A$ can ``trick'' $P_i$ to act non-privately. 
More formally, we model the adversary 
as an interactive randomized functionality.
For a party $P_i$, define $A_{P_i}$ to be a randomized functionality as follows.  \begin{enumerate}
    \item An input of $A_{P_i}$ consists of a sequence of $\ell_i$ entries taken from the data domain, $\mathbf{x}\in X^{\ell_i}$.
    \item $A_{P_i}$ simulates an interaction between party $P_i$ with $\mathbf{x}$ as its input, and $A$. 
    The simulated $P_i$ interacts with $A$ following the instructions of  its part in the protocol, $\Pi_{P_i}$. The adversary $A$ interacts with $P_i$ sending messages for parties $P_1,\ldots,P_{i-1},P_{i+1},\ldots,P_n$. However, $A$ does not necessarily adhere to the protocol $\Pi$.
    \item The simulation continues until $A$ halts with an output $\mathbf{o}_A$, at which time $A_{P_i}$ halts and outputs $\mathbf{o}_A$. \end{enumerate}

\begin{definition}[Multiparty differential privacy \cite{DMNS06,KLNRS11,BeimelNO08,Vadhan2016}]
\label{def:DP_Prot}
A protocol $\Pi$ is $(\epsilon,\delta)$-differentially private if for all $i\in [n]$, for all interactive randomized functionalities $A$, and all neighboring $\mathbf{x},\mathbf{x}'\in X^{\ell_i}$,
$ A_{P_i}(\mathbf{x}) \approx_{\epsilon,\delta} A_{P_i}(\mathbf{x}').$
When $\ell_1=\ell_2=\cdots=\ell_n=1$ we say that the protocol operates in the {\em local model}, and when $n=1$ we say that the protocol (or the algorithm) operates in the {\em curator model}. We say that a protocol $\Pi$ is $\epsilon$-differentially private if  it is $(\epsilon,0)$-differentially private.
\end{definition}

\paragraphnegative{Comparison to previous definitions.} 
In contrast to \cite{BeimelNO08,Vadhan2016}, our definition applies also to a malicious adversary that can send arbitrary messages.
The definition of \cite{KLNRS11} also applies to a malicious adversary, however it requires that each answer of an agent preserves $\epsilon$-differential privacy (i.e., if there are $d$ rounds, the protocol is $d\epsilon$-differentially private). In contrast, the definitions of \cite{BeimelNO08,Vadhan2016} and our definition measures the privacy of the entire transcript of an agent.

Note that (i) Restricting the outcome $\mathbf{o}_A$ to binary results in a definition that is equivalent to \cref{def:DP_Prot}.
(ii) It is possible to consider a relaxed version of \cref{def:DP_Prot} where the adversary $A$ is ``semi-honest'' by requiring $A$ to follow the protocol $\Pi$.
(iii) \cref{def:DP_Prot} differs from definitions of security in the setting of secure multiparty computation as the latter also state correctness requirements with respect to the outcome of the protocol. The difference between the setting is that secure multiparty computation implements a specified functionality\footnote{Furthermore, secure multiparty computation is silent with respect to the chosen functionality, regardless whether it is ``privacy preserving'' or ``secure''.} whereas differential privacy limits the functionality to hide information specific to individuals, but does not specify it.

\subsection{The hybrid model}
\label{sec:hybrid-defn}

A computation in the $(m, n)$-hybrid model is defined as the execution of a randomized interactive protocol $\Pi = (\Pi_C,\Pi_{P_1},\ldots,\Pi_{P_n},\Pi_R)$ between three types of parties: a (single) curator $C$, $n$ ``local model'' agents $P_1,\ldots,P_n$, and a referee $R$. The referee has no input, the curator $C$ receives $m$ input points $\mathbf{x}=(x_1,\ldots,x_m)\in X^m$ and the $n$ ``local model'' agents $P_1,\ldots,P_n$ each receive a single input point $y_{i}\in X$. We use the notation $\mathbf{D}$ to denote the joint input to the computation, i.e., $\mathbf{D}= (x_1,\ldots, x_m,y_1,\ldots,y_n)$. 

The communication in a hybrid model protocol is restricted to messages exchanged between the referee $R$ and the other parties $C,P_1,\ldots,P_n$ (i.e., parties $C, P_1,\ldots,P_n$ cannot directly communicate among themselves). Parties $C,P_1,\ldots,P_n$ have no output, whereas when the execution halts  the referee $R$ writes to a special output register $\mathbf{o}$. See \cref{fig:model}. We require the protocol $\Pi=(\Pi_C,\Pi_{P_1},\ldots,\Pi_{P_n},\Pi_R)$ to satisfy differential privacy as in \cref{def:DP_Prot}.

The hybrid model is a natural extension of well-studied models in differential privacy. Setting $n=0$ we get the trusted curator model (as $C$ can perform any differentially private computation), and setting $m=0$ we get the local model. In this work, we are interested in the case $0 < m \ll n$, because in this regime, the hybrid model is closest in nature to the local model. Furthermore, in many applications, once $m$ is comparable to $n$ it is possible to drop parties $P_1,\ldots,P_n$ from the protocol without a significant loss in utility.

Comparing with Blender~\cite{blender}, where the curator $C$ and the referee $R$ are the same party, we observe that the models are equivalent in their computation power -- every differentially private computation in one model is possible in the other (however, the models may differ in the  number of interaction rounds needed).
Nevertheless, the separation between the curator and the referee has merits as we now discuss.

\paragraphnegative{On the separation between the curator and referee.} 
From a theory point of view, it is useful to separate these two parties as this allows to examine effects related to the order of interaction between the parties (e.g., whether the referee communicates first with the curator $C$ or with the local model parties $P_1,\ldots,P_n$).

Moreover, by separating the curator and referee, the hybrid model encapsulates a richer class of trust models than \cite{blender}, and, in particular, includes a trust model where data sent to the curator is not revealed to the referee. In an implementation this may correspond to a curator instantiated by a publicly trusted party, or by using technologies such as secure multiparty computation, or secure cryptographic hardware which protects data and computation from external processes~\cite{mironov19}.

The curator-referee separation also makes sense from a practical point of view within a company. It is reasonable that only a small fraction of a company's employees, with appropriate clearance and training, should be able to access the raw data of those who contribute their data to the trusted curator model, whereas the majority of employees should only see the privacy-preserving version of it~\cite{rogaway2015moral}.

\begin{remark}[A note on public randomness]
\label{rem:shared}
Some of our protocols assume the existence of a shared random string. In an implementation, shared randomness can be either set up offline or be chosen and broadcast by the referee. We stress that the privacy of our protocols does not depend on the shared random string actually being random. 
Furthermore, all our lower bounds hold even when the local agents hold a shared (public) random string. 
\end{remark}

\subsection{Parity and threshold functions}
A {\em concept} $c:X\rightarrow \{0,1\}$ is a predicate that labels {\em examples} taken from the domain $X$ by either 0 or 1.  A \emph{concept class} $C$ over $X$ is a set of concepts (predicates) mapping $X$ to $\{0,1\}$. 
Let $b,c\in\N$ be parameters. The following two concept classes will appear throughout the paper: 
\begin{itemize}[leftmargin=15pt]
    \item $\Threshold_b=\left\{\Thr_t : t\in\{0,1\}^b\right\}$ where $\Thr_t:\{0,1\}^b\rightarrow\{0,1\}$ is defined as 
$\Thr_t(x)=\1_{\{x\geq t\}}$, where we treat strings from
$\{0,1\}^b$ as integers in $\set{0,\ldots,2^b-1}$.
\item $\Parity_c=\left\{\Par_k : k\in\{0,1\}^c\right\}$ where 
$\Par_k:\{0,1\}^c\rightarrow\{0,1\}$ is defined as 
$\Par_k(x)= \myangle {k,x}=\oplus_{j=1}^c k_j \cdot x_j$.
\end{itemize}

\subsection{Preliminaries from learning theory and private learning}\label{sec:prelim-lt-pl}

We now define the probably approximately correct (PAC) model of~\cite{Valiant84}.
Given a collection of labeled examples, the goal of a learning algorithm (or protocol) is to {\em generalize} the given data into a concept (called a ``hypothesis'') that accurately predicts the labels of fresh examples from the underlying distribution. More formally:

\begin{definition}
The {\em generalization error} of a hypothesis $h:X\rightarrow\{0,1\}$ w.r.t.\ a target concept $c$ and a distribution $\DDD$ is defined as 
$\error_{\DDD}(c,h)=\Pr_{x \sim \DDD}[h(x)\neq c(x)].$ 
\end{definition}

\begin{definition}[PAC Learning~\cite{Valiant84}]\label{def:PAC}
Let $C$ be a concept class over a domain $X$, and let $\Pi$ be a protocol in which the input of every party is a collection of (1 or more) labeled examples from $X$.
The protocol $\Pi$ is called an {\em $(\alpha,\beta)$-PAC learner} for $C$ if the following holds for all concepts $c \in C$ and all distributions $\DDD$ on $X$: 
If the inputs of the parties are sampled i.i.d.\ from $\DDD$ and labeled by $c$, then, with probability at least $1-\beta$, the outcome of the protocol is a hypothesis $h:X\rightarrow\{0,1\}$ satisfying
$\error_{\DDD}(c,h)  \leq \alpha$.

The {\em sample complexity} of the protocol is the total number of labeled examples it operates on. That is, if there are $n$ parties where party $P_i$ gets as input $\ell_i$ labeled examples, then the sample complexity of the protocol is $\ell_1+\dots+\ell_n$.
\end{definition}

A PAC learning protocol that is restricted to only output hypotheses that are themselves in the class $C$ is called a {\em proper learner}; otherwise, it is called an {\em improper learner}. 
A common technique for constructing PAC learners is to guarantee that the resulting hypothesis $h$ has small {\em empirical error} (as defined below), and then to argue that such an $h$ must also have small generalization error.

\begin{definition}
The {\em empirical error} of a hypothesis $h:X\rightarrow\{0,1\}$ w.r.t.\ a labeled sample $\db=(x_i,y_i)_{i=1}^m$ is defined as 
$\error_S(h) = \frac{1}{m} |\{i : h(x_i) \neq y_i\}|.$ 
The {\em empirical error} of $h$ w.r.t.\ an unlabeled sample $\db=(x_i)_{i=1}^m$ and a concept $c$ is defined as $\error_S(h,c) = \frac{1}{m} |\{i : h(x_i) \neq c(x_i)\}|.$ 
\end{definition}

Indeed, in some of our constructions we will use building blocks that aim to minimize the {\em empirical} error, as follows.

\begin{definition}
Let $c:X\rightarrow\{0,1\}$ be a concept and let $\mathbf{D}\in(X\times\{0,1\})^n$  be a labeled database. We say that $\mathbf{D}$ is {\em consistent} with $c$ if for every $(x,y)\in\mathbf{D}$ it holds that $y=c(x)$.
\end{definition}

\begin{definition}\label{def:empLearn}
Let $C$ be a concept class over a domain $X$, and let $\Pi$ be a protocol in which the input of every party is a collection of (1 or more) labeled examples from $X$.
The protocol $\Pi$ is called an {\em $(\alpha,\beta)$-empirical learner} for $C$ if for every concept $c \in C$ and for every joint input to the protocol $\mathbf{D}$ that is consistent with $c$, with probability at least $1-\beta$, the outcome of the protocol is a hypothesis $h:X\rightarrow\{0,1\}$ satisfying
$\error_{\mathbf{D}}(h)  \leq \alpha$.
\end{definition}

We will be interested in PAC-learning protocols that are also differentially private. Specifically,

\begin{definition}[\cite{KLNRS11}]
A {\em private learner} is a protocol $\Pi$ that satisfies both \cref{def:PAC,def:DP_Prot}. Similarly, a {\em private empirical learner} is a protocol $\Pi$ that satisfies both \cref{def:empLearn,def:DP_Prot}.
\end{definition}

Dwork et al.~\cite{DworkFHPRR15} and Bassily et al.~\cite{BassilyNSSSU16} showed that if a hypothesis $h$ is the result of a differentially private computation on a random sample, then the empirical error of $h$ and its generalization error are guaranteed to be close. We will use the following multiplicative variant of their result~\cite{NSunpublished}, whose proof is a variant of the original proof of~\cite{BassilyNSSSU16}.

\begin{theorem}[\cite{DworkFHPRR15,BassilyNSSSU16,NSunpublished,FeldmanS17}] \label{thm:dpGeneralization}
Let $\AAA:X^n\rightarrow2^X$ be an $(\eps,\delta)$-differentially private algorithm that operates on a database of size $n$ and outputs a predicate $h:X\rightarrow\{0,1\}$. Let $\DDD$ be a distribution over $X$, let $S$ be a database containing $n$ i.i.d.\ elements from $\DDD$, and let $h\leftarrow\AAA(S)$. Then,
$$
\Pr_{\substack{S\sim\DDD \\ h\leftarrow \AAA(S)}}\left[ e^{-2\eps} \cdot h(\DDD)-h(S) >  \frac{10}{\eps n}\log\left(\frac{1}{\eps \delta n}\right) \right] < O\left(\eps \delta n\right).
$$
\end{theorem}

We next state known impossibility results for privately learning threshold and parity function.

   \begin{fact} [\cite{BNS13, FeldmanX14}]
    \label{fact:learning_thresholds}
    Let $b\in\N$. 
    Any $\epsilon$-differentially private $(\alpha,\beta)$-PAC learner for  $\Threshold_b$ requires $\Omega(\frac {b}{\epsilon\alpha})$ many samples.
    \end{fact}

    \begin{fact}[\cite{KLNRS11}]
    \label{fact:learning_parity_LDP}
    Let $c\in\N$. 
    In any $\epsilon$-differentially private  $(\alpha,\beta)$-PAC learning protocol for $\Parity_c$ \emph{in the local model} the number of messages is $\Omega(2^{c/3})$. This holds even when the underlying distribution is restricted to be the uniform distribution.
    \end{fact}

\cref{fact:learning_parity_LDP} implies, for example, that when there are $\poly(c)$ agents the number of rounds is  $2^{\Omega(c)}$. It is open if there exists an $\epsilon$-private protocol 
(or an $(\epsilon,\delta)$-private protocol)
for learning $\Parity_c$ in the local model with $\poly(c)$ agents and any number of rounds.

\begin{remark}
The proof of \cref{fact:learning_parity_LDP} in \cite{KLNRS11} is stated in a weaker model, where 
in each round the referee sends an $\epsilon_i$-differentially private local randomizer to an agent and the agent sends the output of this randomizer on its input to the referee, such that
$\epsilon_1+\cdots+\epsilon_\ell \leq \epsilon$. However,
in their proof they only use the fact that $\epsilon_i \leq \epsilon$ in every round, thus, their lower bound proof also applies  to our model.
\end{remark}

Our protocols use the private learner of \cite{KLNRS11} for parity functions, a protocol of \cite{BNST17} for answering all threshold queries,  a  protocol of~\cite{BNS17} for heavy hitters, and a 
 protocol of~\cite{GaboardRS19} for approximating a quantile. These are specified in the following theorems.

\begin{theorem}[\cite{BNST17}]\label{thm:SanThresh}
Let $\alpha,\beta,\eps\leq1$, and let $b\in\N$. There exists a non-interactive $\epsilon$-differentially private protocol in the local model with  $n=O\left(\frac{b^3}{\alpha^2\eps^2}\cdot\log\left(\frac{b}{\alpha\beta\epsilon}\right)\right)$ agents  in which the input of every agent is a single element from $\{0,1\}^b$ and the outcome is a function $q:\{0,1\}^b\rightarrow[0,1]$ such that  for every joint input to the protocol $\mathbf{D}\in(\{0,1\}^b)^n$, with probability at least $1-\beta$, the outcome $q$ is such that $\forall w\in\{0,1\}^b$ we have $q(w)\in\left|\left\{x\in\mathbf{D}: x\leq w\right\}\right|/|\mathbf{D}|\pm\alpha$.
\end{theorem}

\begin{remark}
\cref{thm:SanThresh} does not appear explicitly in \cite{BNST17}, but it is implicit in their analysis. In more details, Bassily et al.~\cite{BNST17} presented a protocol, named {\tt TreeHist}, for identifying {\em heavy hitters} in the input database $\mathbf{D}\in(\{0,1\}^b)^n$. {\tt TreeHist} works by privately estimating for each possible prefix $p\in\{0,1\}^{b'}$ (for $b'\leq b$) the number of input items that agree on the prefix $p$. Once these estimated counts are computed, \cite{BNST17} simply identified the input items $p\in\{0,1\}^{b}$ with large multiplicities in the data. \cref{thm:SanThresh} is obtained from the protocol {\tt TreeHist} (with basically the same analysis) by observing that these estimated counts (for every possible prefix) in fact give estimations for the number of input items within any given {\em interval}. 
This observation has been used several times in the literature, see, e.g., \cite{DworkNPR10}.
\end{remark}

\begin{theorem}[\cite{KLNRS11}]\label{thm:parityLearner}
Let $\alpha,\beta,\eps\leq1$, and let $c\in\N$. There exists an $\epsilon$-differentially private algorithm in the curator model that $(\alpha,\beta)$-PAC learns and $(\alpha,\beta)$-empirically learns $\Parity_c$ properly with sample complexity $O\left(\frac{c}{\alpha\eps}\log(\frac{1}{\beta})\right)$.
\end{theorem}

In fact, the algorithm of \cite{KLNRS11} privately produces a hypothesis with small error (w.h.p.) for every fixed input sample that is consistent with some parity function. 
\begin{theorem}[Heavy hitters protocol~\cite{BNS17}]
\label{thm:heavy-hitters}
There exist constants $\lambda_1,\lambda_2>0$ such that the following holds. 
Let $\beta,\epsilon \leq 1$ and $X$ be some finite domain. There exists a non-interactive $\epsilon$-differentially private protocol in the local model with $n$ agents 
in which the input of each agent is a single element from $X$ and the outcome is a list $\text{\rm Est}$ of elements from $X$ 
such that for every joint input to the protocol $\mathbf{D}\in X^n$, with probability at least 
$1 - \beta$,  every $x$ that is an input of at least  $\frac{\lambda_1}{\epsilon} \sqrt{n \log\left(\frac{|X|}{\beta}\right)}$ agents appears in $\text{\rm Est}$, and vice versa, every element $x$ in $\text{\rm Est}$ is an input of at least
$\frac{\lambda_2}{\epsilon} \sqrt{n \log\left(\frac{|X|}{\beta}\right)}$ agents.

\begin{theorem}[\protect{\cite[Theorem 17]{GaboardRS19}}]
\label{thm:quantiles}
    Let ${\cal P}$ be any distribution on the real line. Fix any $p^*\in (0,1)$ and let $Q_{\min},Q_{\max},q^*$ be such that 
    $\Pr_{x\sim {\cal P}}[x\leq q^*]=p^*$  and $q^* \in [Q_{\min},Q_{max}]$. For any $\epsilon>0$ and for any $\lambda_{\rm quant},\tau_{\rm dist},\beta_{conf} \in (0,\nicefrac 1 2)$, there exists an interactive protocol in the local model with $T = \lceil\log_2(\frac{Q_{\max}-Q_{\min}}{\tau_{\rm dist}})\rceil$ rounds 
    that takes $N$ i.i.d.\ draws from ${\cal P}$, where $N\geq \tfrac {8T} {\lambda_{\rm quant}^2} \left(\frac{e^\epsilon+1}{e^\epsilon-1}\right)^2 \log(\nicefrac {4T}\beta_{conf})$, and 
    with probability at least $1-\beta_{\rm conf}$ it returns $\tilde q$ such that either $|\tilde q-q^*|\leq \tau_{\rm dist}$ or the probability mass ${\cal P}$ places in-between $\tilde q$ and $q^*$ is at most $\lambda_{\rm quant}$.
    \end{theorem}

\end{theorem}
\section{Learning parity XOR threshold}\label{sec:parityXORthresh}
\newcommand{\ParityThresh}{\mathsf{ParityThresh}}
    
\label{sec:learnParityXorThreshold}

In this section we present a learning task that cannot be solved privately in the curator model or in the local model, but can be solved in the hybrid model (without interaction).
The task we consider in this section -- parity XOR threshold -- is similar to the simpler task of the direct product of parity and threshold discussed in~\cref{sec:parity_concatenated_threshold}.  
In this section we design a non-interactive protocol in the
hybrid model for the parity XOR threshold task,
which is more involved than the trivial protocol for the  parity and threshold task. This demonstrates that non-interactive protocols in the hybrid model may have more power than one might initially suspects.

Fix $b,c>0$, and let $k\in\{0,1\}^c$ and $t\in\{0,1\}^b$ be parameters. Define the function $f^{k,t}_{b,c}:\{0,1\}^c\times\{0,1\}^b\rightarrow\{0,1\}$ as follows: $f^{k,t}_{b,c}(x,y)=
\Par_k(x)\oplus\Thr_t(y) = \langle k,x\rangle\oplus\1_{\{y\geq t\}}$ $\bigl($recall that we treat strings in $\set{0,1}^b$ as integers in $\set{0,1,\ldots,2^b-1}$$\bigr)$. Define the concept class
$\ParityThresh$ as follows:
$$\ParityThresh_{b,c}=\left\{f^{k,t}_{b,c} \; : \; k\in\{0,1\}^c \text{ and } t\in\{0,1\}^b \right\}.$$

We first show that every differentially private algorithm (even in the curator model) for learning $\ParityThresh$ must have sample complexity $\Omega(b)$. 

\begin{lemma}
\label{claim:PatityThreshCurator}
Every $\epsilon$-differentially private algorithm for $(\frac{1}{4},\frac{1}{4})$-PAC learning $\ParityThresh_{b,c}$ must have sample complexity $\Omega(b)$.
\end{lemma}

\begin{proof}
Let $\AAA$ be an $\epsilon$-differentially private algorithm for $(\frac{1}{4},\frac{1}{4})$-PAC learning $\ParityThresh_{b,c}$ with sample complexity $m$. We now use $\AAA$ to construct an $\epsilon$-differentially private algorithm $\BBB$ for $(\frac{1}{4},\frac{1}{4})$-PAC learning $\Threshold_b$ with the same sample complexity. By \cref{fact:learning_thresholds} this will show that $m=\Omega(b)$.

Algorithm $\BBB$ is simple: it takes an input database $S=\{(y_i,\sigma_i)\}_{i=1}^m\in\left(\{0,1\}^b\times\{0,1\}\right)^m$ and runs $\AAA$ on the database $\hat{S}=\left\{(\vec{0},y_i,\sigma_i)\right\}_{i=1}^m\in\left(\{0,1\}^c\times\{0,1\}^b\times\{0,1\}\right)^m$ to obtain a hypothesis $\hat{h}$. Then, algorithm $\BBB$ returns the hypothesis $h$ defined as $h(y)=\hat{h}(\vec{0},y)$. Note that changing one element of $S$ changes exactly one element of $\hat{S}$, and hence algorithm $\BBB$ is $\epsilon$-differentially private. 

We next show that algorithm $\BBB$ is a $(\frac{1}{4},\frac{1}{4})$-PAC learner for $\Threshold_b$. To that end, fix a 
target distribution $\DDD$ on $\{0,1\}^b$ and fix a target concept $\Thr_t$ (where $t\in\{0,1\}^b$). Suppose that $S$ contains i.i.d.\ samples from $\DDD$ that are labeled by $\Thr_t$, and consider the following distribution $\hat{\DDD}$: To sample an element from $\hat{\DDD}$ we sample $y\sim\DDD$ and return $\vec{0}\circ y\in \{0,1\}^{b+c}$. Now observe that $\hat{S}$ contains $m$ i.i.d.\ samples from $\hat{\DDD}$ which are labeled by $f_{b,c}^{\vec{0},t}\in\ParityThresh_{b,c}$. Therefore, by the utility properties of $\AAA$, with probability at least $3/4$ the hypothesis $\hat{h}$ satisfies $\error_{\hat{\DDD}}\left(\hat{h},f_{b,c}^{\vec{0},t}\right)\leq\frac{1}{4}$. In that case,
\begin{align*}
    \frac{1}{4} \leq \error_{\hat{\DDD}}\left(\hat{h},f_{b,c}^{\vec{0},t}\right)
    =\Pr_{y\sim\DDD}\left[ \hat{h}\left(\vec{0},y\right)\neq f_{b,c}^{\vec{0},t}\left(\vec{0},y\right) \right]
    =\Pr_{y\sim\DDD}\left[ h\left(y\right)\neq \Thr_t(y) \right]
    =\error_{\DDD}(h,\Thr_t).
\end{align*}
This shows that $\BBB$ is a $(\frac{1}{4},\frac{1}{4})$-PAC learner for $\Threshold_b$, as required.
\end{proof}

We next show that no protocol in the local model can learn $\ParityThresh$, unless the number of exchanged messages is very large.

\begin{lemma}
\label{claim:PatityThreshLocal}
In every $\epsilon$-differentially private protocol in the local model for $(\frac{1}{4},\frac{1}{4})$-PAC learning\linebreak $\ParityThresh_{b,c}$ the number of messages is $\Omega(2^{c/3})$.
\end{lemma} 

The proof of \cref{claim:PatityThreshLocal} is analogous to the proof of \cref{claim:PatityThreshCurator} (using \cref{fact:learning_parity_LDP} instead of \cref{fact:learning_thresholds}). 

So, privately learning $\ParityThresh_{b,c}$ in the curator model requires $\Omega(b)$ labeled examples, and privately learning it in the local model requires $\Omega(2^{c/3})$ messages. We now show that $\ParityThresh_{b,c}$ can be learned privately by a non-interactive protocol in the hybrid model with roughly $O(c)$ examples for the curator and with roughly $O(b^3)$ local agents. We will focus on the case where $c\ll b$. Recall that a function $f^{k,t}_{b,c}(x,y)\in\ParityThresh_{b,c}$ is defined as 
$f^{k,t}_{b,c}(x,y)=\Par_k(x)\oplus\Thr_t(y)$. 
The difficulty in learning $\ParityThresh$ in the hybrid model is that we could only learn the threshold part of the target function using the local agents (since if $c\ll b$ then the curator does not have enough data to learn it), but the threshold label is ``hidden'' from the local agents (because it is ``masked'' by the parity bit that the local agents cannot learn). This false intuition might lead to the design of an {\em interactive} protocol, in which the referee first obtains some information from the curator and then passes this information to the local agents, which would allow them to learn the threshold part of the target function. We now show that such an interaction is not needed, and design a {\em non-interactive} protocol in which the local agents and the curator communicate with the referee only once, simultaneously.

\begin{lemma}\label{lem:PatityThreshHybrid}
There exists a non-interactive $\epsilon$-differentially private protocol in the $(m,n)$-hybrid model for $(\alpha,\beta)$-PAC learning $\ParityThresh_{b,c}$ where $m=O\left(\frac{c}{\alpha^4 \eps}\log(\frac{1}{{\alpha}\beta})\right)$ and $n=O\left(\frac{b^3}{\alpha^4\eps^2}\cdot\log\left(\frac{b}{\alpha\beta\epsilon}\right)\right)$.
\end{lemma}

\begin{proof}
We begin by describing a non-interactive protocol $\Pi$.
The (joint) input to the protocol is a database $\mathbf{D}$ where every point in $\mathbf{D}$ is of the form $(x_i,y_i,\sigma_i)\in\{0,1\}^c\times\{0,1\}^b\times\{0,1\}$. 
At a high level, the protocol works by using the local agents to obtain an approximation to the CDF of the (marginal) distribution on the $y_i$'s (this approximation is given to the referee). In addition, the trusted curator solves $1/\alpha$ parity leaning problems. In more details, the trusted curator sorts its database according to the $y_i$'s, divides its database into $1/\alpha$ chunks, and then applies a private learner for parity functions on each of the chunks. The trusted curator sends the referee the resulting $1/\alpha$ parity functions. The referee then defines the final hypothesis $h$ that, given a point $(x,y)$, first uses the approximation to the CDF (obtained fro the local agents) to match this input point to one of the chunks, and then uses the parity function obtained for that chunk from the trusted curator to predict the label of the point. 

The key observation here is that the threshold part of the target function is {\em constant} on all but at most one of the chunks defined by the trusted curator. As we show, applying a learner for parity on such a ``consistent chunk'' results in a good predictor for the labels of elements of that chunk. Hence, provided that the approximation for the CDF of the $y_i$'s is accurate enough, this results in an accurate learner for $\ParityThresh$. We now formally present the protocol $\Pi$.

\begin{itemize}
    \item  The local agents on a (distributed) input $D=\left(x_i,y_i,\sigma_i\right)_{i=1}^n\in \left(\{0,1\}^c\times\{0,1\}^b\times\{0,1\}\right)^n$:\\
    Run the protocol from \cref{thm:SanThresh} on the (distributed) database $\hat{D}=(y_1,y_2,\dots,y_n)$ with privacy parameter $\eps$ and utility parameters $\alpha^2,\beta$. At the end of the execution, the referee obtains a function $q:\{0,1\}^b\rightarrow[0,1]$ that approximates all threshold queries w.r.t.\ $\hat{D}$.
    
    \item The curator on input $S=\left(x_i,y_i,\sigma_i\right)_{i=1}^m\in \left(\{0,1\}^c\times\{0,1\}^b\times\{0,1\}\right)^m$:
    \begin{itemize}
        \item Sort $S$ according to the $y_i$'s in non-decreasing order.
        \item Divide $S$ into blocks of size $\alpha m$: $S_1,S_2,\dots,S_{1/\alpha}$. For $\ell\in[1/\alpha]$ we denote $S_{\ell}=(x_{\ell,i},y_{\ell,i},\sigma_{\ell,i})_{i=1}^{\alpha m}$.
        \item For every $\ell\in[1/\alpha]$, apply an $\alpha\eps$-differentially private $(\alpha^2,\alpha\beta)$-PAC learner for $\Parity$ on the database $\hat{S}_\ell=(x_{\ell,i}\circ1,\sigma_{\ell,i})_{i=1}^{\alpha m}\in\left(\{0,1\}^{c+1}\times\{0,1\}\right)^{\alpha m}$ to obtain a vector $k_{\ell}\in\{0,1\}^{c+1}$ (using \cref{thm:parityLearner}). 
        \item Send $k_1,\dots,k_{1/\alpha}$ to the referee.
    \end{itemize}
    \item The referee:
    \begin{itemize}
        \item Obtain the function $q$ and the vectors $k_1,\dots,k_{1/\alpha}$.
        \item Define a hypothesis $h:\{0,1\}^c\times\{0,1\}^b\rightarrow\{0,1\}$ as  $h(x,y)=\langle x\circ1,k_{I(y)}\rangle$, where $I(y)=\left\lceil \frac{q(y)}{\alpha}\right\rceil$.
        \item Output $h$.
    \end{itemize}
\end{itemize}

The privacy properties of the protocol $\Pi$ are straightforward, as both the local agents and the curator apply $\eps$-differentially private computations: The local agents apply the algorithm from \cref{thm:SanThresh}, and the curator applies an $\alpha\eps$-differentially private computation on each of the blocks $S_1,\dots,S_{1/\alpha}$ (note that changing one element of $S$ can change at most one element of each of these blocks). 

We now proceed with the utility analysis. 
Fix a target function $f_{b,c}^{k^*,t^*}\in\ParityThresh_{b,c}$ and fix a target distribution $\DDD$ on $\{0,1\}^c\times\{0,1\}^b$. We use $\DDD_c$ and $\DDD_b$ to denote the marginal distributions on $\{0,1\}^c$ and $\{0,1\}^b$, respectively. We will make the simplifying assumption that $\DDD_b$ does not give too much weight on any single point in $\{0,1\}^b$, specifically, $\Pr_{w\sim\DDD_b}[w=y]\leq\beta/m^2$ for every $y\in\{0,1\}^b$. This assumption can be enforced by padding every example with $O(\log(m/\beta))$ uniformly random bits.

Let $S$ and $D$ (the inputs to the curator and the local agents) be sampled i.i.d.\ from $\DDD$ and labeled by $f_{b,c}^{k^*,t^*}$. We next show that w.h.p.\ the resulting hypothesis $h$ has low empirical error on $S$. By standard generalization arguments, such an $h$ also has low generalization error. 

First observe that there is at most one index $\ell^*\in[1/\alpha]$ such that $\Thr_{t^*}(y_{\ell^*,1})\neq\Thr_{t^*}(y_{\ell^*,\alpha m})$. In all other blocks $S_\ell$ we have that $\Thr_{t^*}(\cdot)$ is constant on all the $y_{\ell,i}$'s of that block.  We will show that w.h.p.\ the hypothesis $h$ has small empirical error on every such block. Fix $\ell\neq\ell^*$, and let $\nu\in\{0,1\}$ be the value of $\Thr_{t^*}(\cdot)$ on the $y_{\ell,i}$'s of the $\ell$th block (that is, for every $i\in[\alpha m]$ we have $\Thr_{t^*}(y_{\ell,i})=\nu$). Recall that since the elements of $S$ are labeled by $f_{b,c}^{k^*,t^*}$, for every $i\in[\alpha m]$ we have that $$\sigma_{\ell,i}=f_{b,c}^{k^*,t^*}(x_{\ell,i},y_{\ell,i})=\langle k^*,x_{\ell,i} \rangle\oplus\Thr_{t^*}(y_{\ell,i})=\langle k^*,x_{\ell,i} \rangle\oplus\nu=\langle k^*\circ\nu,\;x_{\ell,i}\circ1 \rangle.$$
Hence, the elements of $\hat{S}_{\ell}$ are all labeled by the parity function defined by $k^*\circ\nu$. Therefore, as $k_{\ell}$ is the outcome of the learner form \cref{thm:parityLearner} on $\hat{S}_{\ell}$, for $m\geq O\left(\frac{c}{\alpha^2 \eps}\log(\frac{1}{\alpha\beta})\right)$, with probability at least $1-\alpha\beta$ we have that $\error_{\hat{S}_{\ell}}(\Par_{k_{\ell}})\leq\alpha^2$, 
that is, $\langle k_{\ell},\; x\circ1\rangle$ is a good predictor for the label of the elements in block $S_{\ell}$. 

Recall that the hypothesis $h$ matches inputs $(x,y)$ to the vectors $k_1,\dots,k_{1/\alpha}$ using the function $q$ obtained from the local agents, that is, on input $(x,y)$, the hypothesis uses $k_{\ceil{ q(y)/\alpha}}$. Therefore, to complete the proof we need to show that most of the elements from block $S_{\ell}$ are matched by the hypothesis $h$ to the vector $k_{\ell}$. To that end, let $\#_S(w)=|\{(x,y,\sigma)\in S : y\leq w \}|$, and consider the following event:
$$\text{\bf Event } {\boldsymbol E_1:}\qquad \forall w\in\{0,1\}^b \;\text{ it holds that }\; \left| q(w) - \frac{1}{m}\cdot\#_S(w) \right|\leq 4\alpha^2.$$

We first conclude the proof assuming that Event $E_1$ occurs. Fix $\ell\neq\ell^*$, and recall that the elements of $S$ (and in particular the elements of $S_{\ell}$) are sorted in a non-decreasing order according to their $y_i$'s. Now fix $8\alpha^2 m \leq i \leq \alpha m- 8\alpha^2 m$. 
By our simplifying assumption (that the distribution $\DDD_b$ does not put a lot of mass on any single point), we may assume that all the $y_i$'s in $S$ are distinct, which happens with probability at least $1-\beta$. In that case, we have that $\#_S(y_{\ell,i})=\max\{0, \ell-1\}\cdot\alpha m +i$, and hence, 
$$\max\{0, \ell-1\}\cdot\alpha +8\alpha^2 \leq
\frac{1}{m}\#_S(y_{\ell,i}) \leq \max\{0, \ell-1\}\cdot\alpha + \alpha - 8\alpha^2.$$
By Event $E_1$ we get that
$$\max\{0, \ell-1\}\cdot\alpha +4\alpha^2 \leq
q(y_{\ell,i}) \leq \max\{0, \ell-1\}\cdot\alpha + \alpha - 4\alpha^2,$$
and so, $\left\lceil \frac{q(y_{\ell,i})}{\alpha}\right\rceil = \ell$. That is, for all but at most $16\alpha^2 m$ elements of the block $S_{\ell}$ we get that $h(x_{\ell,i},y_{\ell,i})=\langle x_{\ell,i}\circ1,\;k_{\ell}\rangle=\Par_{k_{\ell}(x_{\ell,i},y_{\ell,i})}$. Recall that $\Par_{k_{\ell}}$ errs on at most $\alpha^2 m$ elements of $S_{\ell}$, and so the hypothesis $h$ errs on at most $17\alpha^2 m$ elements of the block $S_{\ell}$. That is, $h$ errs on at most $17\alpha^2 m$ elements of every block $S_{\ell}$ for $\ell\neq\ell^*$, and might err on all of $S_{\ell^*}$ which is of size $\alpha m$. So, $h$ errs on at most $\frac{1}{\alpha}\cdot17\alpha^2 m + \alpha m=18\alpha m$ elements of $S$. Standard generalization bounds now state that, except with probability at most $\beta$, we get that $\error_{\DDD}(h,f_{b,c}^{k^*,t^*})\leq O(\alpha)$ (in particular, this follows from the generalization properties of differential privacy; see \cref{sec:prelim-lt-pl} for more details). Overall, with probability at least $1-O(\beta)$ the resulting hypothesis has generalization error at most $O(\alpha)$.

It remains to show that Event $E_1$ occurs with high probability. First,
for $n\geq O\left(\frac{b^3}{\alpha^4\eps^2}\cdot\log\left(\frac{b}{\alpha\beta\epsilon}\right)\right)$, \cref{thm:SanThresh} ensures that with probability at least $1-\beta$ the function $q$ is such that $\forall w\in\{0,1\}^b$ it holds that $\left|q(w)-\frac{1}{n}\#_{\hat{D}}(w)\right|\leq\alpha^2$, where $\#_{\hat{D}}(w)=|\{y\in \hat{D} : y\leq w \}|$. Second, by standard generalization arguments, assuming that $n$ and $m$ are big enough, we would also have that $\frac{1}{n}\#_{\hat{D}}(w)$ and $\frac{1}{m}\#_S(w)$ are both within $\alpha^2$ from $\Pr_{y\sim\DDD_b}[y\leq w]$. Specifically, by the Dvoretzky-Kiefer-Wolfowitz inequality~\cite{dvoretzky1956asymptotic,massart1990tight}, assuming that $n$ and $m$ are at least $\Omega\left(\frac{1}{\alpha^4}\log(\frac{1}{\beta})\right)$, this happens with probability at least $1-\beta$. Assuming that this is the case, by the triangle inequality we have that Event $E_1$ holds. This shows that Event $E_1$ happens with probability at least $1-3\beta$, and completes the proof.
\end{proof}

We remark that it is possible to design a more efficient learner for $\ParityThresh$ (in terms of sample complexity) by constructing a protocol in which there are multiple rounds of communication between the referee and the local agents (but this communication is still independent from the message that the curator sends to the referee). This will be illustrated in \cref{sec:parity_concatenated_threshold}. 
We summarize our possibility and impossibility results w.r.t.\ learning $\ParityThresh$ in the next theorem (which follows from \cref{claim:PatityThreshCurator} and \cref{claim:PatityThreshLocal} and from \cref{lem:PatityThreshHybrid}).

\begin{theorem}
\label{thm:parityXORthreshold}
Let $c\in\N$ and $b=c^2$. Then there is a non-interactive $\frac{1}{4}$-differentially private $(\frac{1}{4},\frac{1}{4})$-PAC learner for $\ParityThresh_{b,c}$ in the $(m,n)$-hybrid model with $m=O(c)$ samples for the curator and $n=O(c^6 \log c)$ local agents. However, every such learner in the local model with $o(2^{(n/\log n)^{1/6}})$ local agents requires $2^{\Omega((n/\log n)^{1/6})}$ rounds, and every such learner in the curator model requires $\Omega(m^2)$ samples.
\end{theorem}

\section{The 1-out-of-$2^d$-parity task}
\label{sec:1-out-of-2d-parity}

In this section we describe a  task that cannot be privately solved neither in the curator model nor  in the local model with sub-exponential number of rounds. In the hybrid model,  this task can only be solved with interaction, first with the local agents and then with the curator. In this task there are many instances of the parity problem and the referee needs to solve only one instance, which is determined by the inputs. The local agents can determine this instance (using a heavy hitters protocol) and the curator can now solve this instance. The curator cannot solve all instances since this will exceed its privacy budget, and by the definition of the task the curator will not have enough information to determine the instance; thus interaction with both the local agents and the curator is required. 

\begin{definition}[The 1-out-of-$2^d$-parity task]
The inputs in the 1-out-of-$2^d$-parity task  are generated as follows: \begin{enumerate}
    \item {\bf Input:} $2^d$ strings $(r_j)_{j \in \bit^d}$, where $r_j \in \bit^c$ for every $j\in \bit^d$, and  $m+1$  elements $s_1,\dots,s_{m+1}\in \bit^d$.
    \item Set $s = s_1\oplus\cdots \oplus s_{m+1}$.%
    \footnote{The strings $s_1,\dots,s_{m+1}$ are an $m+1$-out-of-$m+1$ secret sharing of $s$, that is, together they determine $s$, but every subset of them gives no information on $s$.}
    \item \label{item:1-out-of-parity-generate} 
    Each sample $x_1,\dots,x_m$ and $y_1,\dots,y_n$ is generated  independently as follows:
        \begin{itemize}
            \item with probability half choose $x \inr \bit^c$ with uniform distribution and output $(x,(\myangle{x,r_j})_{j \in \bit^d})$
            (that is, every point contains a string $x$ of length $c$ and $2^d$ bits which are the inner products of $x$ and each of the $r_j$'s).
            \item with probability half choose $t \inr [m+1]$ with uniform distribution  and output $(t,s_{t})$
            (that is, every point contains a number $t$  
            and the $t$-th   string $s_t$).
        \end{itemize}
    \end{enumerate}
The goal of the referee in the 1-out-of-$2^d$-parity task is for every $(r_j)_{j \in \bit^d}$ and $s_1,\dots,s_{m+1}$ to recover $r_s$ with probability at least $1-\beta$, where the probability is over 
the generation of the inputs in \stepref{item:1-out-of-parity-generate} and
the randomness of the parties in the protocol.
\end{definition}

We start by describing a protocol for this task.

\begin{lemma}
\label{lem:1out2^dparityProtocol}
Let $\beta > 1/m$ and assume that $m=\Omega\left(\frac{c\log(1/\beta)}{\epsilon}\right)$ and $n =\Omega \left(\frac{m^2}{\epsilon^2}\log(\frac{m 2^d}{\beta})\right)$. 
The 1-out-of-$2^d$-parity task can be solved in the $(m,n)$-hybrid model by an $\epsilon$-differentially private protocol with three rounds, where in the first round each local agent sends one message to the referee (without seeing any other messages), in the second round the referee  sends one message to the curator, and in the third round the curator sends one message to the referee.
\end{lemma}
\begin{proof}
The protocol is as follows:
In the first round the local agents send messages according to the $\epsilon$-differentially private heavy hitters protocol of~\cref{thm:heavy-hitters} (from~\cite{BNS17}) with the inputs $(t,s_{t})$ and $\beta/3$, that is, a protocol that returns with probability at least $1-\beta/3$ all values that are inputs of ``many'' agents.  
If the input of a local agent is not $(t,s_t)$ for some $t$, then it executes the protocol with some default input $\bot$. The referee reconstructs the $m+1$ strings $s_1,\dots,s_{m+1}$ (with probability at least $1-2\beta/3$), computes $s=s_1\oplus \cdots \oplus s_{m+1}$, and sends it to the curator. The curator privately solves the parity task with inputs $(x,\myangle{x,r_s})$ using the algorithm  of~\cref{thm:parityLearner} (from ~\cite{KLNRS11}) with $\alpha=1/4$ and $\beta/3$. 
Since we use  $\epsilon$-differentially private algorithms, each operating on different inputs, the resulting protocol is $\epsilon$-differentially private.

We next argue that with probability
at least $1-\beta$, the referee reconstructs $r_s$. 
Note that for a fixed $t \in [m+1]$, the expected number of times that  $(t,s_t)$ is an input of agents $P_1,\dots,P_n$ is
$n/2(m+1)$. By a simple Chernoff bound, with probability $1-\beta/3$ for all $t$, the value $(t,s_t)$ is an input of at least $n/4(m+1)$ parties. 
By \cref{thm:heavy-hitters},
with probability at least $1-\beta/3$, each value that is an input of at least $O\left(1/\epsilon \sqrt{n \log(\frac{m 2^d}{\beta/3})}\right)$ agents will appear in the list computed by the referee.
By the assertion of the lemma, $O\left(1/\epsilon \sqrt{n \log(\frac{m 2^d}{\beta/3}) }\right) < \frac{n}{4(m+1)}$. Thus, with probability at least $1-2\beta/3$, the referee reconstructs all $s_t$'s and reconstructs the correct value $s$.  Furthermore, as $m =\Omega(\frac{c\log(3/\beta)}{\epsilon})$,
the algorithm of~\cref{thm:parityLearner} returns, with probability at least $1-\beta/3$, a string $r$ such that $\Pr[\myangle{x,r} \neq \myangle{x,r_s}] \leq 1/4$ under the uniform distribution on $x \in \bit^c$. Since for $r\neq r_s$ this probability is exactly $1/2$, we get that $r=r_s$ with probability at least $1-\beta$.
\end{proof}

We next prove that, unless the database is big,  the 1-out-of-$2^d$-parity task requires interaction. To prove this result, we first convert a  protocol for the 1-out-of-$2^d$-parity task to a private algorithm in the trusted curator model that recovers all strings $(r_j)_{j \in \bit^d}$.
We then prove, using a simple packing argument, that, unless the database is ``big'',  such algorithm cannot exist. 
For our proof, we define the all-$2^d$-parity task as the task in which all inputs are of the form $(x,(\myangle{x,r_j})_{j \in \bit^d})$
and the goal of the referee is to reconstruct all strings $(r_j)_{j \in \bit^d}$. 
\begin{claim}
\label{cl:1out2^dparityTransformation}
Let $m <n$.
If there is an  $\epsilon$-differentially private protocol for the
$1$-out-of-$2^d$-parity problem  in the $(m,n)$-hybrid model in which the curator and the referee can exchange many messages and then the referee
simultaneously sends one message to each local agent and gets one answer from each agent, then there is an $\epsilon$-differentially private algorithm in the trusted curator model for the all-$2^d$-parity problem for a database of size $O(nd)$.
\end{claim}
\begin{proof}
Let $\Pi$ be an $\epsilon$-differentially private protocol with the above interaction pattern for the
$1$-out-of-$2^d$-parity task in the $(m,n)$-hybrid model in which the referee reconstructs $r_s$ with probability at least $1-\beta$. We construct,  in three steps,  an algorithm $\calA$ for the all-$2^d$-parity task in the trusted curator.

First, we construct from $\Pi$ a protocol $\Pi'$ in the $(O(md) ,O(dn))$-hybrid model that reconstructs $r_s$ with error probability at most $\beta/2^d$
(e.g., execute $\Pi$ with disjoint inputs $O(d)$ times and take the value $r_s$ that is returned in the majority of the executions). 

Next, we construct from $\Pi'$ a  protocol $\Pi''$ for the the
all-$2^d$-parity task in the $(O(md),O(nd))$-hybrid model (with error probability at most $\beta$). 
In $\Pi''$, the parties holding inputs of the all-$2^d$-parity problem simulate $\Pi'$ on inputs for the $1$-out-of-$2^d$-parity task as follows:
\begin{itemize}
    \item The curator on input $(x_i,(\myangle{x_i,r_j})_{j\in \bit^d})_{i=1}^m$:
    \begin{itemize}
        \item Chooses random $s_1,\dots,s_{m+1} \inr \bit^d$.
        \item For each $i \in [m]$, with probability $1/2$ replaces its $i$-th input by $(t_i,s_{t_i})$ for a uniformly distributed $t_i\inr [m+1]$.
        \item Exchanges messages with the referee as specified by $\Pi'$ on the new input. In addition it sends to the referee  $s_1,\dots,s_{m+1}$ and an index $\ell$ such that $(\ell,s_{\ell})$ does not appear in its new input.
    \end{itemize}
    \item The referee after getting the message from the curator:
    \begin{itemize}
        \item Chooses a set $A \subseteq [n]$ with uniform distribution.
        \item For every $i \notin A$, sends its message in $\Pi'$ to the $i$-th agent and gets an answer $\msg_i$ from the agent.
        \item For every $i \in A$, chooses a random $q_i \inr [m+1]$. Let $B=\set{i\in A: q_i =\ell}$.
        \item For every $i \in A \setminus B$, computes (without any interaction) its message in $\Pi'$  to agent $P_i$ and the answer $\msg_i$ of agent $P_i$ with input $(q_i,s_{q_i})$.
        \item For every $i \in B$ and $s \in \bit^d$, computes (without any interaction) its message in $\Pi'$ to agent $P_i$ and the answer $\msg_{i,s}$ of agent $P_i$ with input $(\ell,s\oplus \bigoplus_{k \neq \ell} s_{k})$.
        \item For every $s \in \bit^d$, reconstructs $r_s$ from the messages of the curator in $\Pi'$, $(\msg_i)_{i \notin B}$, and $(\msg_{i,s})_{i\in B}$.
    \end{itemize}
\end{itemize}

As the curator holds $m$ samples and there are $m+1$ values $s_1,\dots,s_{m+1}$, there exists an index $\ell$ such that $(\ell,s_\ell)$ does not appear in the new input of the curator. Thus,  the referee for every $s \in \bit^d$ can choose a value $s'_\ell$ such that it is consistent with the messages of the curator in $\Pi'$ and $s=s'_\ell \oplus \bigoplus_{k \neq \ell} s_k$. Furthermore, each of $x_1,\dots,x_m,y_1,\dots,y_n$ is replaced with probability half with a value $(t,s_t)$ for a uniformly distributed $t$, thus, these inputs are distributed as required for the 1-out-of-$2^d$-parity task.  This implies that for every $s\in \bit^d$ the referee reconstructs $r_s$ from the messages of the curator in $\Pi'$, $(\msg_i)_{i \notin B}$, and $(\msg_{i,s})_{i\in B}$ with probability at least $1-\beta/2^d$. By the union bound, the referee correctly reconstructs all $(r_j)_{j \in \bit^d}$ with probability at least $1-\beta$.

Finally, we construct the desired algorithm $\calA$ from $\Pi''$. The trusted curator simply simulates the referee, the curator, and the agent in $\Pi''$, that is, it take its database with $O((m+n)d)$ samples and partitions it to $(x_1,\dots,x_{O(md)})$ (the input of the curator) and $y_1,\dots,y_{O(nd)}$, computes without any interaction a random transcript of $\Pi''$ on these inputs, and reconstructs the output $(r_j)_{j\in \bit^d}$. Since the transcript preserves $\epsilon$-differential privacy and computing the output is post-processing, algorithm $\calA$ is $\epsilon$-differential private.
\end{proof}

\begin{claim}
\label{cl:2^dparityLB}
If there is exists an $\epsilon$-differentially private algorithm in the trusted curator model for
the all-$2^d$-parity task with strings of length $c$, then $n = \Omega\left(\frac{c2^d+\ln(1-\beta)}{\epsilon}\right)$.
\end{claim}
\begin{proof}
The proof is by a simple packing argument. For every  strings $(r_j)_{j\in \bit^d}$, with probability at least $1-\beta$, the algorithm returns $(r_j)_{j\in \bit^d}$ when the samples are generated with $(r_j)_{j\in \bit^d}$. By the group privacy of $\epsilon$-differential privacy, with probability at least $e^{-n\epsilon}(1-\beta)$ the algorithm returns $(r_j)_{j\in \bit^d}$ when the samples are generated with $(0^c)_{j\in \bit^d}$. As there are $2^{c2^d}$ options for $(r_j)_{j\in \bit^d}$ and the above events are disjoint, $2^{c2^d}e^{-n\epsilon}(1-\beta) \leq 1$, i.e., $n = \Omega\left(\frac{c2^d+\ln(1-\beta)}{\epsilon}\right)$.
\end{proof}

\begin{lemma}
\label{lem:1out2^dparityLB}
Let $m <n$.
If there is an  $\epsilon$-differentially private protocol for the
$1$-out-of-$2^d$-parity task  in the $(m,n)$-hybrid model with $\beta=1/4$ in which the curator and the referee can exchange many messages and then the referee
simultaneously sends one message to each local agent and gets one answer from each agent, then $n =\Omega(c 2^d/d \epsilon)$.
\end{lemma}
\begin{proof}
By  \cref{cl:1out2^dparityTransformation}, if there exists an $\epsilon$-differentially protocol in the $(m,n)$-hybrid model for the $1$-out-of-$2^d$-parity task, then 
there exists an $\epsilon$-differentially private algorithm in trusted curator model for the all-$2^d$-parity task with database of size $O(dn)$.
Thus, by \cref{cl:2^dparityLB} with $\beta=1/4$, 
$dn = \Omega( \frac{c2^d}{\epsilon})$.
\end{proof}

\cref{lem:1out2^dparityLB} is valid also if the local  agents are allowed to hold a shared (public) random string as this string can be sent by the referee to each agent as part of its message (without adding extra rounds of communication).

We summarize the possibility and impossibility results for the $1$-out-of-$2^d$-parity task in the following theorem, where, for convenience, we choose specific parameters that highlight these results. 

\begin{theorem}
\label{thm:1-out-of-2^d parity}
Let $\epsilon=1/4$, $\beta=1/4$. For every integer $c$, there are $d=\Theta(\log c)$, $m=\Theta(c)$, and 
$n =\Theta(c^2\log c)$ such that
\begin{enumerate}
    \item \label{item:protocol}
    There exists an $\epsilon$-differentially private protocol for the $1$-out-of-$2^d$-parity task with strings of length $c$ in the $(m,n)$-hybrid model where first each local agent sends one message to the referee and then the referee exchanges one message with the curator.
    \item \label{item:non-ineractive} There does not exist an $\epsilon$-differentially private protocol for this task in the $(m,n)$-hybrid model in which the referee first exchanges messages with the curator  and then simultaneously exchanges one message with the local agents.
    \item \label{item:localParity} In any  $\epsilon$-differentially private protocol for this task in the local model with $n$ agents the number of rounds is $2^{\Omega(c)}=2^{\Omega(\sqrt{n/\log n})}$.
    \item \label{item:curator} There is no algorithm in the trusted curator model that solves this task with $m$ examples.
\end{enumerate}
\end{theorem}
\begin{proof}
\cref{item:protocol} follows directly from \cref{lem:1out2^dparityProtocol}. 
For \cref{item:non-ineractive}, by \cref{lem:1out2^dparityLB},
with $\epsilon=1/4$, $\beta=1/4$, and $d=2\log c+\log\log c$
$$n=\Omega\left(\frac{c 2^d}{d\epsilon}\right)=\Omega(c^3),$$
contradicting the choice of $n=\Theta(c^2\log c).$

For the impossibility result in \cref{item:localParity}, recall that by \cref{fact:learning_parity_LDP} the  number of messages sent to the referee in an $\epsilon$-differentially private learning protocol in the local model for parity of strings of length $c$ with respect to the uniform distribution is $2^{\Omega(c)}$.
By simple simulation, an $\epsilon$-differentially private  protocol in the local model for the  $1$-out-of-$2^d$-parity task implies an $\epsilon$-differentially private  protocol in the local model for learning parity with respect to the uniform distribution (with the same number of messages). Specifically, since the number of agents is $n=O(c^2 \log c)$, the number of rounds is $2^{\Omega(c)}/(c^2 \log c)=2^{\Omega(\sqrt{n/\log n})}$.

For \cref{item:curator}, observe that a curator receiving $m$ input points obtains less than $m+1$ shares of $s$ and hence obtains no information about $r_s$. Hence, such a curator cannot solve the 1-out-of-$2^m$-parity task alone, even without privacy constraints.
\end{proof}

\section{The parity-chooses-secret task}

\label{sec:paritySecretSharing}

We now present another task that cannot be privately solved neither in the curator model nor in the local model with sub-exponential number of rounds.  This task can be solved in the hybrid model; however, it requires interaction, this time first with the curator and then with the local agents. This task (as well as 1-out-of-$2^d$-parity task) highlights both the information and private-computation gaps between the curator and the local model agents. The local model agents receive enough information to solve the task, but lack the ability to privately solve an essential sub-task. The curator does not receive enough information to solve the task (even non-privately), however the curator can be used to privately solve the hard sub-task. Once the hard sub-task is solved, this information is forwarded to the local agents, which now can solve the task. 

\begin{definition}[The parity-chooses-secret task]
The inputs in the parity-chooses-secret task are generated as follows:
\begin{enumerate}
    \item {\bf Input:} A string $r\in \bit^c$ and 
    $2^{c}$ vectors of $m+1$ bits: a vector  $(s_{j,1},\dots,s_{j,m+1})\in\bit^{m+1}$
    for every $j\in \bit^c$. 
    \item Set $s_j = s_{j,1}\oplus\cdots \oplus s_{j,m+1}$ for every for $j\in \bit^c$, i.e., $s_j$ is a random bit shared via an $m+1$-out-of-$m+1$ secret-sharing scheme, with the shares being $s_{j,1},\dots,s_{j,m+1}$.
    \item \label{item:parity-chooses-secret-generate} 
    Each sample $x_1,\dots,x_m$ and $y_1,\dots,y_n$ is generated independently as follows:
        \begin{itemize}
            \item Choose $x \inr \bit^c$ and $t\inr [m+1]$ and output 
             $(x,\myangle{x,r},t,(s_{j,t})_{j \in \bit^c})$ 
             (that is, every point contains a string of length $c$, its inner product with $r$, an integer $t$, and the $t$-th share of each $s_j$). 
    \end{itemize}
\end{enumerate}
The goal of the referee in the  parity-chooses-secret task is for 
every $r$ and every
      $\bigl((s_{j,1},\dots,s_{j,m+1})\bigr)_{j \in \bit^c}$ to recover $s_r$ with probability at least $1-\beta$, where the probability is over the generation of the inputs in \stepref{item:parity-chooses-secret-generate} and the randomness of the parties.
\end{definition}

We start by describing a protocol for the parity-chooses-secret task.

\begin{lemma}
\label{lem:ParityChoosesSecretProtocol}
Let $\beta > 1/m$ and assume that $m=\Omega\left(\frac{c\log(1/\beta)}{\epsilon}\right)$ and $n =\Omega\left(\frac{m^2}{\epsilon^2}\log(\frac{n}{\epsilon})\right)$.
The parity-chooses-secret task can be solved in the $(m,n)$-hybrid model by an $\epsilon$-differentially private protocol with three rounds, where in the first round the curator sends one message to the referee, in the second round the referee  sends one message to the local agents, and in the third round each local agent sends one message to the referee.
\end{lemma}
\begin{proof}
The protocol is as follows:
First, the curator learns $r$  by executing the $\epsilon$-differentially private algorithm for learning parity of~\cite{KLNRS11} (see \cref{thm:parityLearner})  with  $\alpha=1/4$ , $\beta/3$, and the  $m$ inputs $(x,\myangle{x,r})$. The curator  sends $r$ to the referee, which forwards it to the local agents. Next each local agent sends a messages according to the $\epsilon$-differentially private heavy hitters protocol of~\cite{BNS17} (see \cref{thm:heavy-hitters}) with the input $(t,s_{r,t})$ and $\beta/3$.  The referee recovers 
$(1,s_{r,1}),\dots,(m+1,s_{r,m+1})$ from the messages of the heavy-hitters protocol, and outputs  y$s_r = s_{r,1}\oplus\cdots \oplus s_{r,m+1}$.
Since we use  $\epsilon$-differentially private algorithms, each operating on different inputs, the resulting protocol is $\epsilon$-differentially private.

We next argue that with probability
at least $1-\beta$, the referee reconstructs $s_r$.
As $m =\Omega(\frac{c\log(1/\beta)}{\epsilon})$,
the algorithm of~\cite{KLNRS11} (see \cref{thm:parityLearner}) returns, with probability at least $1-\beta/3$, a string $r'$ such that $\Pr[\myangle{x,r'} \neq \myangle{x,r}] \leq 1/4$ under the uniform distribution on $x \in \bit^c$. Since for $r\neq r'$ this probability is exactly $1/2$, we get that $r=r'$ with probability at least $1-\beta/3$.
We complete the proof by showing that, once the local agents and the referee know $r$, the referee 
reconstructs with probability at least $2\beta/3$ all values $(1,s_{r,1}),\dots,(m+1,s_{r,m+1})$. 
Note that for a fixed $t \in [m+1]$, the expected number of times that  $(t,s_{r,t})$ is an input of agents $P_1,\dots,P_n$ is
$n/(m+1)$. By a simple Chernoff bound, with probability $1-\beta/3$, for all $t$ the value $(t,s_{r,t})$ is an input of at least $n/2(m+1)$ parties. The protocol of~\cite{BNS17} (see \cref{thm:heavy-hitters}) guarantees that, with probability at least $1-\beta/3$, each value that is an input of at least $O\left(1/\epsilon \sqrt{n \log(\frac{n}{\beta/3})}\right)$ agents will appear in the list computed by the referee. By the assertion of the lemma, $O\left(1/\epsilon \sqrt{n \log(\frac{n}{\beta/3})}\right) < \frac{n}{2(m+1)}$. Thus, with probability at least $1-\beta$, the referee reconstructs  $s_{r,1},\dots,s_{r,m+1}$ and reconstructs the correct value $s_r$.  \end{proof}

We next prove that, unless the database is big,  the parity-chooses-secret task requires interaction. Furthermore, we rule-out  protocols in which first
the referee  simultaneously sends one message to each local agent, then receives an answer from each local agent, and finally exchanges (possibly many)  messages with the curator. In particular, we rule-out
the communication pattern used in \cref{lem:1out2^dparityProtocol} for the 
1-out-of-$2^d$-parity task.
To prove this result, we first convert a  protocol $\Pi$ for the parity-chooses-secret task with the above communication pattern to
a  protocol $\Pi'$ in the hybrid model with the same communication pattern for a similar task (which we call the parity-chooses-secret' task, defined below). 
We then convert the protocol $\Pi'$ to a non-interactive protocol $\Pi''$ in the local model for another related task, and complete the proof 
by showing an impossibility result for the related task. 

We define the parity-chooses-secret'  task as the task in which the input of the curator is generated as in 
the parity-chooses-secret  task and the 
input of each local agent only contains  shares, that is, it is of the form $(t,(s_{j,t})_{j \in \bit^c})$. The goal of the referee in remains the same -- to recover $s_r$.
\begin{claim}
\label{cl:parity-chooses-secret'}
Assume that $m=\Omega\left(\frac{c\log(1/\beta)}{\epsilon}\right)$.
If there is an $\epsilon$-differentially private protocol for the parity-chooses-secret task in the $(m,n)$-hybrid model with error at most $\beta$
in which in the first round the referee simultaneously sends one message to each local agent, in the second round gets an answer from each agent, and then the referee and the curator exchange (possibly many) messages,  then there is a  $2\epsilon$-differentially private protocol for the parity-chooses-secret' task in the $(m,n)$-hybrid model
with error at most $2\beta$ with the same communication pattern. 
\end{claim}
\begin{proof}
Let $\Pi$ be an   $\epsilon$-differentially private protocol   for the
parity-chooses-secret task in the $(m,n)$-hybrid model with the communication pattern as in the claim in which the referee reconstructs $s_r$ with probability at least $1-\beta$. We construct from $\Pi$ a   $2\epsilon$-differentially private protocol $\Pi'$ with the same communication pattern for the the
parity-chooses-secret' task in which the referee reconstructs $s_r$ with probability at least $1-2\beta$.
In $\Pi'$, each agent $P_i$, holding an input $(t,(s_{j,t})_{j \in \bit^c})$, chooses with uniform distribution a string $x_i \inr \bit^c$ and sends two  messages of $\Pi$, one message,
denoted $\msg_{i,0}$, for the input
$(x_i,0,t,(s_{j,t})_{j \in \bit^c})$ and one message,
denoted $\msg_{i,1}$,
for the input
$(x_i,1,t,(s_{j,t})_{j \in \bit^c})$. In addition, the agent sends $x_i$ to the referee. 
The referee sends the messages that it gets from the local agents (i.e., $(x_i,\msg_{i,0},\msg_{i,1})_{i \in [n]})$ and its random string to the curator. The curator does as follows:
\begin{itemize}
\item
Privately learns $r$  by executing the $\epsilon$-differentially private algorithm of~\cite{KLNRS11} (see \cref{thm:parityLearner}) for learning parity with 
$\alpha=1/4$, $\beta$, and the
$m$ inputs $(x,\myangle{x,r})$.
\item
For each agent $P_i$, computes $b_i=\myangle{x_i,r}$ 
and $\msg_{i}=\msg_{i,b_i}$
(that is, the curator chooses the correct message from the two messages the agent sends).
\item
Simulates the communication between the curator and the referee
in $\Pi$ assuming that the curator gets the messages
$(\msg_{i})_{i\in [n]}$ in the first round and 
 reconstructs $s_r$ as the referee reconstructs it in $\Pi$.
\item
Sends $s_r$ to the referee. 
\end{itemize}
As each party executes two $\epsilon$-differentially private algorithm on its input, the resulting protocol is $2 \epsilon$-differentially private.
Each agent $P_i$ chooses $x_i$  with uniform distribution (as in the  parity-chooses-secret task). Furthermore, as $m$ is big enough, with probability  at least $1-\beta$, the curator computes the correct value $r$. Thus, $(x_i,b_i,t,(s_{j,t})_{j \in \bit^c})$ is an input distributed as required in the parity-chooses-secret task, and the curator reconstructs $s_r$ with probability at most $1-2\beta$.
\end{proof}

We recall a result of~\cite{MMPRTV10} showing that the mutual information between an the input and output of  a differential private algorithm is low. Recall that the {\em entropy} $H(X)$ of a random variable $X$ is defined as 
$$H(X) \triangleq -\sum_{x, \Pr[X=x]>0} \Pr [X=x]\log \Pr [X=x].$$
It can be proved that
$  0 \leq H(X) \leq \log(|\supp(X)|),$ where $|\supp(X)|$ is the
size of the support of $X$ (the number of values with probability greater
than zero). The upper bound $\setsize{\supp(X)}$ is obtained if and only if
the distribution of $X$ is uniform and the lower bound is obtained if and
only if $X$ is deterministic.
Given two  random variables $X$ and $Y$ (possibly dependent), the
{\em conditioned entropy} of $X$ given $Y$ is defined as $H(X|Y) \triangleq H(XY) -
H(Y).$
From the definition
of the conditional entropy, the following properties can be proved:
$$ H(XY) \leq H(X)+H(Y),
$$
and for $3$ random variables $X,Y,Z$
$$
H(X|YZ) \leq H(X|Y).
$$
The {\em mutual information} between $X$ and $Y$ is defined as 
$$I(X;Y)  \triangleq
H(X) - H(X|Y)=H(X)+H(Y)-H(XY).$$

\begin{theorem}[Differential privacy implies low mutual information~\cite{MMPRTV10}] 
\label{thm:MI}
Let $A : X^n
\rightarrow Y$ be an
$\epsilon$-differentially private mechanism. Then for every random variable $V$ distributed on $X^n$, we
have
$I(V; A(V)) \leq 1.5\epsilon n$.
\end{theorem}

\begin{lemma}
\label{lem:lbPCS}
Assume that $m=\Omega\left(\frac{c\log(1/\beta)}{\epsilon}\right)$.
If there is an $\epsilon$-differentially private protocol for the parity-chooses-secret task in the $(m,n)$-hybrid model with error at most $\beta$ 
in which in the first round the referee simultaneously sends one message to each local agent, in the second round gets an answer from each agent, and then the referee and the curator exchange (possibly many) messages, then 
$n \geq \frac{(1-2\beta)2^c}{3 \epsilon}$.
\end{lemma}
\begin{proof}
We convert the protocol for the parity-chooses-secret task to a protocol $\Pi''$ in the local model with $n$ agents, where the input of each agent contains $2^c$ bits $(s_j)_{j\in \bit^c}$. If the inputs of all agents are equal, then for every $r\in \bit^c$ the referee should output the bit $s_r$ with probability at least $1-\beta$. We will show at the end of the proof that such protocol can exist only if $n$ is big.

By \cref{cl:parity-chooses-secret'}, under the assumption of the lemma
there is a  $2\epsilon$-differentially private protocol $\Pi'$ for the parity-chooses-secret' task in the $(m,n)$-hybrid model
with error at most $2\beta$ and communication pattern is as in the lemma. We construct the following protocol $\Pi''$ in the local model with $n$ agents: 
\begin{itemize}
    \item {\bf Input of each agent $\boldsymbol{P_i}$:} $(s_j)_{j\in \bit^c}$.
    \item 
    The referee chooses with uniform distribution $2^c$ vectors of $m+1$ bits: a vector  $(s_{j,1},\dots,s_{j,m+1})\inr\bit^{m+1}$
    for every $j\in \bit^c$.
    \item 
    The referee chooses with uniform distribution $m$ indices  $t_1,\ldots,t_m \inr [m+1]^m$. Let $\ell$ be an index that does not appear in this list.
    \item 
    The referee chooses with uniform distribution $m$ strings  $(x_1,\ldots,x_m) \inr (\bit^c)^m$. 
    \item
    The referee sends $((s_{j,1},\dots,s_{j,m+1}))_{j \in \bit^c}$ and $\ell$ to each agent.
    \item
    Each agent $P_i$ chooses with uniform distribution an index $t \in [m+1]$. If $t \neq \ell$, it sends to the referee its message in $\Pi'$ on input $t,(s_{j,t})_{j \in \bit^c}$.   If $t = \ell$, it sends to the referee its message in protocol $\Pi'$ on input $\ell,\left(s_j\oplus \bigoplus_{k \neq \ell} s_{j,k}\right)_{j \in \bit^c}$. Denote the message of $P_i$ by $\msg_i$. 
    \item 
    For every $r \in \bit^c$, the referee does the following:
    \begin{itemize}
        \item Computes (without interaction) the communication exchanged in $\Pi'$ between the referee and the curator with input $$(x_1,\myangle{x_1,r},t_1,(s_{j,t_1})_{j\in \bit^c}),\ldots,(x_m,\myangle{x_m,r},t_m,(s_{j,t_m})_{j\in \bit^c}).$$ Denote this communication by $\msg_{C,r}$. 
        \item The referee reconstructs $s_r$ from the messages $\msg_{C,r},\msg_1,\dots,\msg_n$ using the reconstruction function of $\Pi'$.
    \end{itemize}
\end{itemize}

Protocol  $\Pi''$ is $2\epsilon$-differentially private, since $\Pi'$ is $2\epsilon$-differentially private. Furthermore, if all the inputs of the local agents are equal, then $\msg_{C,r},\msg_1,\dots,\msg_n$ are distributed as in $\Pi'$, thus, for every $r \in \bit^c$,  the referee reconstructs $s_r$ with probability at least $1-\beta$. 

We complete the proof by showing that $n$ must be large enough in $\Pi''$ (hence, also in $\Pi$). 
Assume we execute protocol $\Pi''$ when $(s_j)_{j \in \bit^c}$ is choosen with uniform distribution and 
denote its input by $(s_j)_{j \in \bit^c}$  and its output by
$(s'_j)_{j \in \bit^c}$. As the output in $\Pi''$ is computed from the transcript of $\Pi'$, the post-processing  property of differential privacy implies that 
the algorithm that first executes protocol $\Pi'$ and then computes the output from the transcript is $2\epsilon$-differentially private. 
By \cref{thm:MI}, 
\begin{equation}
\label{eq:Ilower}
I\left((s_j)_{j \in \bit^c};(s'_j)_{j \in \bit^c}\right) \leq 3\epsilon n.
\end{equation}
On the other hand, $\Pr[s_{j_0}=s'_{j_0}] \geq 1-2\beta$ for a given $j_0 \in \bit^c$, thus
$$
H(s_{j_0}|(s'_j)_{j \in \bit^c})\leq 
H(s_{j_0}|s'_{j_0}) \leq 2\beta,$$
and 
$$ 
H\left((s_j)_{j \in \bit^c}|(s'_j)_{j \in \bit^c}\right)
\leq \sum_{j_0 \in \bit^c} H\left(s_{j_0}|(s'_j)_{j \in \bit^c}\right) \leq 2\beta 2^c.
$$
Thus, 
\begin{equation}
\label{eq:IGreater}
I\left((s_j)_{j \in \bit^c};(s'_j)_{j \in \bit^c}\right) = H\left((s_j)_{j \in \bit^c}\right)
- H\left((s_j)_{j \in \bit^c}|(s'_j)_{j \in \bit^c}\right) 
\geq 2^c -2\beta 2^c =(1-2 \beta) 2^c.
\end{equation}
Inequalities (\ref{eq:Ilower}) and (\ref{eq:IGreater}) imply that
$ (1-2 \beta) 2^c \leq 3\epsilon n$.
\end{proof}

We summarize the possibility and impossibility results for the parity-chooses-secret task in the next theorem.

\begin{theorem}
Let $\epsilon=1/4$, $\beta=1/4$. For every integer $c$, there are  $m=\Theta(c)$ and 
$n =\Theta(c^2\log c)$ such that 
\begin{enumerate}
    \item \label{item:protocolPCS}
    There exists an $\epsilon$-differentially private protocol for the parity-chooses-secret task with strings of length $c$ in the $(m,n)$-hybrid model where first the curator sends one message to the referee and then the referee simultaneously exchanges one message with each local agent.
    \item \label{item:non-ineractivePCS} There does not exist an $\epsilon$-differentially private  protocol for this task in the $(m,n)$-hybrid model in which the referee first simultaneously exchanges one message with the local agents and then exchanges messages with the curator.
    \item \label{item:localParityPCS} In any  $\epsilon$-differentially private protocol for this task in the local model with $n$ agents the number of rounds is $2^{\Omega(c)}$.
    \item \label{item:curatorPCS} There is no algorithm in the  trusted curator model that solves this task with $m$ examples. 
\end{enumerate}
\end{theorem}
\begin{proof}
\cref{item:protocol} follows directly from \cref{lem:ParityChoosesSecretProtocol}. 
\cref{item:non-ineractivePCS} is implied by \cref{lem:lbPCS}, since $n \ll 2^c$.  
\cref{item:localParityPCS} follows  from \cref{fact:learning_parity_LDP} as in the proof of  \cref{thm:1-out-of-2^d parity}. 

For \cref{item:curatorPCS}, observe that a curator receiving $m$ input points obtains less than $m+1$ shares of $(s_1,\ldots,s_j)$ and hence obtains no information about $s_r$. That is, such a curator cannot solve the parity-chooses-secret task alone, even without privacy constraints.
\end{proof}

\section{A Negative Result: Basic hypothesis testing}
\label{sec:basic_hypothesis_testing}

Here, we show that for one of the most basic tasks, differentiating between two discrete distributions $\calD_0$ and $\calD_1$, the hybrid model gives no significant added power.

\begin{definition}[The simple hypothesis-testing task]
Let $0< \beta<1$ be a parameter, $X$ be a finite domain, and $\calD_0$ and $\calD_1$ be two distributions over $X$.
The  input of the hypothesis-testing task is composed of i.i.d.\ samples from $\calD_j$ for some $j\in\{0,1\}$ and the goal of the referee is to output $\hat j$ s.t.\ $\Pr[\hat j = j]\geq {1-}\beta$.
\end{definition}

\begin{theorem}
\label{thm:no_utility_for_hypothesis-testing}
If there exists an $\epsilon$-differentially private protocol in the $(m,n)$-hybrid model
for  testing  between distributions $\calD_0$ and $\calD_1$ with success probability $1/2+\gamma$, then either there  exists an $\epsilon$-differentially private  protocol for this task in the curator model that uses $m$ samples and succeeds with probability at least $1/2+\gamma/4$ or there  exists an $\epsilon$-differentially private  protocol for this task in the local model with $n$ agents that  succeeds with probability at least $1/2+\gamma/4$.
\end{theorem}
\begin{proof}
Let $\Pi$  be a protocol guaranteed by the lemma, that is,
when the inputs of the curator and the local agents are drawn from $\calD_j$, the referee in $\Pi$ returns $j$ with probability at least $1/2+\gamma$.
Consider an execution of the protocol when the inputs of the curator  are drawn from $\calD_0$ and the inputs of the local agents are drawn from $\calD_1$ and let 
$p$ be the probability that the referee in $\Pi$ returns 1 in this case.

We first assume that $p\geq 1/2$ and show that there exists an $\epsilon$-differentially private protocol 
$\Pi^{\rm local}$ for this task in the local model with $n$ agents that  succeeds with probability at least $1/2+\gamma/4$. The referee in protocol $\Pi^{\rm local}$ with probability $\gamma/2$ outputs 1 and with probability $1-\gamma/2$  draws $m$ samples from $\calD_0$, executes protocol $\Pi$, where the referee simulates the messages of the curator using the $m$ samples, and returns the output of $\Pi$.

We next analyze this protocol. If the inputs of the local agents are drawn
from $\calD_1$, then the probability that the referee in protocol $\Pi$  returns $1$ is at least $1/2$ and the probability that the referee in $\Pi^{\rm local}$ returns $1$ is at least
$\gamma/2+(1-\gamma/2)\cdot 1/2=1/2+\gamma/4$. 
If the inputs of the local agents are drawn
from $\calD_0$, then the probability that the referee in $\Pi^{\rm local}$ returns $0$ is at least
$(1-\gamma/2)\cdot (1/2+\gamma) \geq 1/2+\gamma/4$ (since $\gamma \leq 1/2$).

For the case that $p < 1/2$, it can be shown, using an analogous construction, that there exists an $\epsilon$-differentially private protocol 
$\Pi^{\rm curator}$ for this task in the curator model with $m$ samples  that  succeeds with probability at least $1/2+\gamma/4$.
\end{proof}

\paragraph{Notation.} The \emph{total variation distance} (also known as the statistical distance) of two discrete distributions $\calD_0, \calD_1$ over a domain $X$ is  $d_{\rm TV}(\calD_0, \calD_1) = \sup_{T\subset X} |\calD_1(T)-\calD_0(T)| = \frac{1}{2} \sum_{x\in X} |\calD_1(x)-\calD_0(x)|$. The \emph{squared Hellinger distance} between two distributions $\calD_0, \calD_1$ over a domain $X$ is defined as $d_{H^2}(\calD_0, \calD_1) = \frac 1 2\sum_{x\in X}\left( \sqrt{\calD_0(x)} - \sqrt{\calD_1(x)} \right)^2$.

For the rest of the discussion in this section, fix the domain $X =\{0,1\}$, and some $\alpha>0$. We define two distributions $\calD_0$ and $\calD_1$  
where under $\calD_1$ we have $\Pr_{x\sim \calD_1}[x=1]= \frac 1 2 (1+\alpha)$ and under $\calD_0$ we have $\Pr_{x\sim \calD_0}[x=1]= \frac 1 2 (1-\alpha)$.
It is a fairly simple calculation to see that $d_{\rm TV}(\calD_0, \calD_1)=\alpha$ and $\frac{\alpha^2} 2 \leq d_{H^2}(\calD_0, \calD_1) \leq \alpha^2$.
We prove that for some setting of the parameters $n$ and $m$,
the hypothesis-testing task between $\calD_0$ and $\calD_1$ is
impossible in the $(m,n)$-hybrid model.

Next, we cite two known results regarding differentially private simple hypothesis testing. The work of Joseph et al.~\cite{JosephMNR19} discusses sample complexity bounds for simple hypothesis testing the in local-model.

\begin{fact}[{\cite[Theorem~5.3]{JosephMNR19}}]
\label{fact:LDP_protocols_for_basic_hypothesis_testing}
Let $\Pi$ be an $\epsilon$-differentially private protocol in the local model  for distinguishing between a setting where the input of the $n$ local agents is drawn i.i.d.\ from $\calD_0$ vs.~the input drawn i.i.d.\ from $\calD_1$. Let $\Pi^j$ denote the view of the protocol under $n$ input point drawn i.i.d.\ from $\calD_j$ for $j\in\{0,1\}$. Then $d_{\rm TV}(\Pi^1,\Pi^0) \leq 50n\epsilon^2d_{\rm TV}(\calD_0,\calD_1) = 50n\epsilon^2 \alpha^2$.
\end{fact}

The work of Cannone et al.~\cite{CanonneKMSU19} gives tight sample complexity bounds for private hypothesis testing in the curator model. Although their result for general distributions is rather technical to state, doing so for the result for distributions supported on precisely two elements is fairly simple.

\begin{fact}[{\cite[Theorem~3.5]{CanonneKMSU19} reworded}]
\label{fact:curator_samplecomplexity_basic_hypothesis_testing}
There exists a constant $C>0$ such that any $\epsilon$-DP algorithm for distinguishing w.p. $\geq 0.55$ between a setting where the input of $m$ datapoints is drawn i.i.d.\ from $\calD_0$ vs.~the input drawn i.i.d.\ from $\calD_1$ requires 
\[ m\geq C \left(\frac 1 {d_{H^2}(\calD_0,\calD_1)} + \frac 1 {\epsilon \cdot d_{\rm TV}(\calD_0,\calD_1)}\right) = \Omega(\frac 1 {\alpha^2}+ \frac 1 {\epsilon\cdot \alpha}) \]
\end{fact}

Combining the above two facts with Theorem~\ref{thm:no_utility_for_hypothesis-testing} we obtain the following as an immediate corollary.

\begin{theorem}
\label{thm:no_utility_for_sum}
There exists two constants $c_0, c_1$ such that in the $(m,n)$-hybrid model, with $m\leq c_0\left(\frac {1}{\alpha^2}+\frac 1 {\epsilon\alpha}\right)$ and $n\leq c_1\cdot \frac{1}{\epsilon^2\alpha^2}$, there does not exists an
$\epsilon$-differentially private protocol that succeeds in determining whether all $m+n$ input points are drawn from $\calD_0$ vs.~$\calD_1$ w.p.\ $\geq 0.75$. 
\end{theorem}
\begin{proof}
Assume towards a contradiction that such a protocol $\Pi$ exists.
By  \cref{thm:no_utility_for_hypothesis-testing}, there is an
$\epsilon$-differentially private protocol that succeeds in determining whether all input points are drawn from $\calD_0$ vs.~$\calD_1$ w.p.\ $\geq 0.5625$ either in the trusted curator model with sample complexity $m$
or in the local model with $n$ agents. The former yields an immediate contradiction with Fact~\ref{fact:curator_samplecomplexity_basic_hypothesis_testing} and the latter yields an immediate contradiction with Fact~\ref{fact:LDP_protocols_for_basic_hypothesis_testing} for a sufficiently small $c_1$.
\end{proof}

\section{Acknowledgements}

We thank Adam Smith for suggesting the select-then-estimate task discussed in the introduction.
The work was done while the authors were at the ``Data Privacy: Foundations and Applications'' program held in spring 2019 at the Simons Institute for the Theory of Computing, UC Berkeley. 

Work of A.~B.\ and K.~N.\ was supported by NSF grant No.~1565387 TWC: Large: Collaborative: Computing Over Distributed Sensitive Data.
Work of A.~B.\ was also supported by ISF grant no.~152/17, a grant from the Cyber Security Research Center at Ben-Gurion University, and ERC grant 742754 (project NTSC).

Work of A.~K. was supported by NSF grant No.~1755992 CRII: SaTC: Democratizing Differential Privacy via Algorithms for Hybrid Models, a VMWare fellowship, and a gift from Google.

Work of O.~S.\ was supported by grant \#2017–06701 of the Natural Sciences and Engineering Research Council of Canada
(NSERC). The bulk of this work was done when O.~S.\ was affiliated with the University of Alberta, Canada

Work of U.~S.\ was supported in part by the Israel Science Foundation (grant No.\ 1871/19).

\appendix

\section{Learning parity and threshold}
\label{sec:parity_concatenated_threshold}

We now define a problem which, due to known results regarding sample complexity lower bounds, cannot be privately learned  neither in the curator model nor in the local model with sub-exponential number of rounds and yet can be learned in the hybrid model. 
Specifically, we make use of existing impossibility results for learning threshold functions in the curator model (see \cref{fact:learning_thresholds}) and for learning parity functions in the local model (see \cref{fact:learning_parity_LDP}). The learning problem we define is fairly natural -- it is the concatenation of the two problems, resulting in a two-tuple label. All that is left is to set parameters so that either the curator or the local model agents fail to learn the suitable part of the problem.

For  integers $b,c$, 
we define a concept class $${\cal C}_{ b,c}=\set{c_{k,t}:\set{0,1}^{{c}} \times \set{0,1}^{ b}\rightarrow\set{0,1}^2:k \in \set{0,1}^{{c}}, t \in \set{0,1}^{b} },$$ where $c_{k,t}(x,y) = (\Par_k(x), \Thr_t(y))$.

\begin{theorem}
\label{thm:learning_parity_concatenated_threshold}
Let $b$ be an integer and $0<\epsilon,\alpha,\beta<1$ s.t.\  $8\cdot 10^5 b\log(b/\beta\alpha) / (\alpha^3\epsilon^2) < 2^{\sqrt b/4}$ and let
$c=\sqrt{b}$, $m = 1000\sqrt b\log(1/\beta)/\epsilon\alpha$, and $n = 8\cdot 10^5b\log(b/\beta\alpha) / (\alpha^3\epsilon^2)$.
The task of PAC-learning ${\cal C}_{ b,c}$ with $\epsilon$-differentially privacy can be solved in the $(m,n)$-hybrid model 
yet cannot be learned  neither in the curator model with $m$ samples nor in the local agents model with $n$ agents and at most $2^{\Omega(\sqrt{b})}/n=\frac{\alpha^3 \epsilon^22^{\Omega(\sqrt{b})}}{\log 1/\beta \alpha}$ rounds. 
\end{theorem}
    \begin{proof}
    We begin with the latter part of the theorem. Assume towards contradiction that there exists an $\epsilon$-differentialy private PAC-learner 
    $\AAA$ for ${\cal C}_{b,c}$ in the curator model with a sample size of at most $m$. As we explain below, this implies that there exists an $\epsilon$-differentialy private algorithm $\AAA'$ in the curator model for PAC-learning a \Threshold\ problem that has access to at most $m$ examples, contradicting \cref{fact:learning_thresholds}. The construction of $\AAA'$ from $\AAA$ is as follows:  Algorithm $\AAA'$, given $m$ inputs to the \Threshold\ problem, picks $k\in \{0,1\}^{c}$ arbitrarily and pads each example with a $x\in\{0,1\}^{{c}}$  chosen with uniform distribution and the label $\langle x,k\rangle$ and then feeds the concatenated inputs to $\AAA$. By definition, the $m$ input points are legal inputs to $\AAA$ and thus w.p.\ $\geq 1-\beta$ the algorithm produces a good $h$. Projecting $h$ onto its second coordinate yields an hypothesis $h'$ whose error in comparison to $\Thr_t$ is at most $\alpha$.
    The proof that no differentialy private protocol 
    in the local model  with at most $2^{c/3}/n$ rounds can learn such $h$ on its own is symmetric and follows from   \cref{fact:learning_parity_LDP}.

    We now show that there exists a protocol in the $(m,n)$-hybrid model
        that is capable of privately learning ${\cal C}_{b,c}$. The protocol itself is very straight-forward: the curator $(\tfrac\alpha 2, \tfrac \beta 2)$-learns the parity portion of $h$ (its second coordinate) and the local agents $(\tfrac\alpha 2, \tfrac \beta 2)$-learn the threshold portion ($h$'s first coordinate). In this protocol,
    the interaction of the referee with the curator is independent from 
     the interaction of the referee with the local agents.
     
    We use the \Parity\ learning algorithm of \cref{thm:parityLearner} (from~\cite{KLNRS11}), and, not surprisingly, $m$ is set s.t.\  it sufficiently large to learn a good hypothesis w.p.\ $\geq 1 - \tfrac\beta 2$. The differentialy-private protocol in the local model for learning \Threshold\ is folklore, so for completion we detail it here. Basically, it is composed of two steps: first finding all quantiles of $\{ \frac \alpha 4, \frac{2\alpha} 4, \frac{3\alpha}4,\ldots, 1-\frac{\alpha}4 \}$ of the distribution, and then figuring out in which of the quantiles there is a flip of the labels from $0$ to $1$.

    In more details, here is the differentialy-private protocol in the local model for learning the threshold function. We partition the $n$ agents into two sets. The first set is then further equipartitioned into $\lceil \frac 4 \alpha\rceil -1$ sets, where in set $i$ we find the $i\cdot \frac \alpha 4$-quantile of the distribution over the $y$'s, up to an error of at most $\frac\alpha{10}$ using, say, the $\epsilon$-differentialy-private binary-search protocol for quantiles in the local model of \cref{thm:quantiles} (from~\cite{GaboardRS19}).
       In our setting, with a discrete distribution taking values between $0$ and $2^b$, we set $Q_{\max}=2^b$ and $Q_{\min}=0$, $\tau_{\rm dist}=1/4$,
    $\beta_{\rm conf}=\frac{\beta\alpha}{16}$,
    and $\lambda_{\rm quant} = \frac{\alpha}{10}$ as discussed above; 
    hence a sample of size $10^5\frac{b}{\epsilon^2\alpha^2}\log(8b/\beta_{\rm conf})$ is sufficient for a single quantile approximation w.p. $\geq 1- \beta_{\rm conf}$.
    We have set $n$ such that a sample complexity of $(n/2)/(\lceil 4/\alpha\rceil -1)$ suffices for finding the suitable quantile with an error probability of at most $\frac{\beta\alpha}{16}$. Once all quantiles of the form $i\cdot \frac \alpha 4$ have been published, we turn to the latter half of the $n$ local agents. We apply the heavy-hitters
    algorithm of \cref{thm:heavy-hitters} (from~\cite{BNS17}) for the agents with $1$-label, 
    that is, the input of an agent with $1$-label to the heavy-hitters protocol is its quantile and the input of an agent with $1$-label is 0.  
    Note again that $n/2$ is sufficiently large s.t.\  w.p.\ $\geq 1- \tfrac \beta 4$ we can assess the fraction of $1$-labels in all bins. Furthermore, by definition of the threshold problem, any bin that contains only values $\geq t$ must be all $1$-labeled. The agents then find the first quantile from which all agents have $1$-label and output it as the suggested threshold. By definition, w.p.\ $\geq 1- \frac \beta 2$ we have that the error of this threshold is the worst-case width of any bin, i.e. at most $\frac \alpha 4   + 2\cdot \frac \alpha {10} <\frac \alpha 2$. Thus we have a $\epsilon$-differentialy private protocol in the local model for $(\frac\alpha 2,\frac\beta 2)$-learning thresholds with $n$ agents.
\end{proof}

We remark that it is possible to design a {\em non-interactive} protocol in the hybrid model for PAC learning ${\cal C}_{b,c}$, at the expense of increasing the sample complexity of the protocols. See \cref{sec:parityXORthresh}.

\section{The select-then-estimate task}
\label{sec:selectThenEstimate} 

Select-then-estimate is an example of a task that cannot be privately solved in the curator model or in the local model but can be solved (with interaction) in the hybrid model. We do not know whether the interaction is essential. 

\begin{definition}[The select-then-estimate task] Let $X=\set{-1,1}^d$, and let $P\in\Delta(X)$ be an unknown distribution over $X$. Define $\mu = {\mbox E}_{x\sim P}[x]$. Note that  $\mu=(\mu_1,\ldots,\mu_d)\in [-1,1]^d$. The inputs in the select-then-estimate task are $x_1,\ldots,x_n$ sampled i.i.d.\ from $P$.

For parameters $\alpha < \alpha'$, the goal of the referee in the select-then-estimate task is to output $(i, \hat \mu_i)$ such that $\mu_i\geq \mbox{max}_j (\mu_j)-\alpha$ and $ \left|\hat\mu_i - \mu_i\right| \leq \alpha'$.
\end{definition}

We focus on the case when $\alpha' < \alpha$, that is, the selection $i$ is done with the lesser accuracy $\alpha$, and the estimation $\hat \mu_i$ is with the better accuracy $\alpha'$. 
In the following discussion, we omit the dependency on the allowed failure probability for the task, $\beta$.

A number of sample complexity bounds are relevant for this problem for non-interactive protocols. Selecting a coordinate $i$ such that $\mu_i \geq \max_j \mu_j - \alpha$ requires $\Theta(\frac{\log d}{\alpha^2}+\frac{\log d}{\alpha \epsilon})$ samples in the trusted curator model under pure differential privacy~\cite{McSherryT07, bafna2017price, steinke2017tight}, and $\Theta(\frac{d \log d}{\alpha^2 \epsilon^2})$ samples in the local model~\cite{Ullman18}.\footnote{The bound in~\cite{Ullman18} is for non-interactive local model protocols. In \cref{sec:selectionLowerbound} we show that for interactive local model protocols $n=\Omega(\frac{d}{\alpha\epsilon})$.} Once $i$ is selected, estimating $\mu_i$ up to an error $\alpha'$ (i.e., computing $\hat{\mu_i}$ so that $|\hat{\mu_i} - \mu_i| \leq \alpha'$) requires $\Omega(\frac{1}{\alpha'^2})$ samples regardless of privacy, and $O(\frac{1}{\alpha'^2 \epsilon^2})$ samples suffice in the local model
(using the randomized respond protocol of Warner~\cite{Warner65}).

To exemplify the many settings where the hybrid model provides a solution for the task but neither curator alone nor the local model parties alone can solve it, consider the case where $\epsilon=0.1$, $\alpha <0.1$, and $\alpha/\sqrt{d} < \alpha' < \alpha/\log{d}$. 
With this choice of parameters, $m=O(\log d/\alpha^2)$ samples suffice for the curator to identify $i$, but not to perform the estimate because $m=o(1/\alpha'^2)$ as $1/\alpha'^2 = \Omega(\log^2d/\alpha^2)$. 
Likewise, a choice of $n=O(1/\alpha'^2)$ allows for performing the estimate in the local model, but not for performing the selection, as $1/\alpha'^2 = O(d/\alpha^2)$.

The $(m,n)$-hybrid model with the above parameter choices, however, allows for a solution where the curator first identifies $i$ such that $\mu_i \geq \max_j \mu_j - \alpha$ and then the estimation of $\mu_i$ within accuracy $\alpha'$ is performed by the local model parties.

By our analysis above, for the case $\epsilon=0.1, \alpha < 0.1$ we get that the greater the ratio $n/m$ is, the better the improvement in accuracy as $$\frac{\alpha}{\alpha'} \sim \frac{\sqrt{\log d/m}}{\sqrt{1/n}} = \sqrt{\frac{n \log d}{m}}.$$

Another way to look at it is to get a sense of how small the size of the curator contributors $m$ can be in comparison to $n$ in order to be helpful. It depends on the ratio of desired accuracies $\frac{\alpha'}{\alpha},$ as $\frac{m}{n} \sim (\frac{\alpha'}{\alpha})^2 \log d $, i.e., $\frac{\log d}{d} < \frac{m}{n} < \frac{1}{\log d}.$ 

The analysis can be repeated for other values of $\epsilon$ and $\alpha$ as well. For example, the $(m, n)$ hybrid model can privately solve the select-then-estimate task when $\epsilon = 1, \alpha = \frac{1}{d^{0.25}}, \alpha' = \frac{1}{d^{0.75-0.5t}}, m = \sqrt{d} \log d, n = d^{1.5-t} \log d,$ for any choice of $0.5 < t < 1,$ and neither the local model nor the curator can solve it to the same accuracies separately.

In the high-dimensional problems that are of interest for the select-then-estimate task, when often the dimension of the data $d$ exceeds the number of data points available $n$, the parameter values for when the model is helpful illustrate its considerable potential.

\section{A lower bound for the selection function}
\label{sec:selectionLowerbound}
We next present a simple lower bound of the number of messages sent in an $\epsilon$-differentially private protocol for the selection problem.
\begin{definition}[The selection problem]
The input for the selection problem is $n$ i.i.d.\ samples $Y=(y_1,\dots,y_n)$ from an unknown distribution $\calD$ over $\set{0,1}^d$ with mean $\mu=(\mu_1,\dots,\mu_d)$. The goal is to output a coordinate $j$ such that
$$E_{Y,j}[\mu_j] \geq \max_k \mu_k -\alpha.$$
\end{definition}

Ullman~\cite{Ullman18} proved that in any $\epsilon$-differentially private non-interactive protocol in the local model for the selection problem, the number of agents is $\Omega\left(\frac{d \log d}{\alpha^2 \epsilon^2}\right)$. We prove a lower bound of $\Omega(d^{1/3})$ on the number messages sent to the referee in any $1$-differentially private protocol for this task, even if interaction is allowed. For example, if each agent sends one message, then the number of agents is $\Omega(d^{1/3})$ (even if the protocol is interactive, e.g., the referee sends a message to $P_1$, gets a message from $P_1$, based on this message computes a message and sends it to $P_2$, and so on). As  another example, if there are only $O(d^{1/6})$ agents, then the number of rounds in any $1$-differentially private protocol for the selection problem is $\Omega(d^{1/6})$. To summarize,  our lower bound applied to a larger family of protocols than the result of~\cite{Ullman18}; however, our bound is weaker.

The idea of our proof is simple: We show a reduction from private learning in the local model  to privately solving the selection problem and use a known lower bounds for the parity problem~\cite{KLNRS11} to obtain the lower bound.

\begin{claim}
\label{cl:LerningToSelection}
Let $C=\set{c_1,\dots,c_d}$ be a class of functions
and $n$ be such that $n\geq\frac{64}{\alpha}(\VC(C)\ln(\frac{128}{\alpha})+\ln(8))$.
If there is an $\epsilon$-differentially private protocol $M$ in the local model for the selection problem with $n$ agents, then there is an $\epsilon$-differentially private proper $(4\alpha,1/4)$-PAC learning protocol in the local model for the class $C$ with $n$ agents and the same number of messages.
\end{claim}
\begin{proof}
We convert a labeled example for $C$ to an input for the selection problem: Given a labeled example $(x,b)$, construct the input $\convert(x,b)=(y[1],\dots,y[d])$, where  if $c_i(x)=b$ then $y[i]=1$, else $y[i]=0$.  
Given  $((x_1,b_1),\ldots,(x_n,b_n))$ -- examples labeled by the some concept $c_i$ --  we construct
$Y=(\convert(x_1,b_1),\ldots,\convert(x_n,b_n)),$   execute $M$ on $Y$, and return $c_j$, where $j$ 
is the output of $M$.

Notice that if the samples are labeled by $c_i$, then the $i$-th coordinate in all points in $Y$ is $1$. Thus, $M$ returns a coordinate $j$ such that $E_{Y,j}[\mu_j] \geq 1 -\alpha$; in particular, the probability that $\mu_j$ is less than $1-4\alpha$ is at most $1/4$.  Thus, with probability at least $3/4$, the learning algorithm has error at most $4\alpha$ on the sample. By  standard VC arguments, with probability at least $1/2$ the concept $c_j$ has error at most $8 \alpha$ on the distribution $\calD$.  
\end{proof}

\begin{theorem} Suppose there is a $1$-differentially private local protocol $M$ in the local model for the selection problem  with $\alpha=1/10$. Then, the number of messages sent by the agents to the referee 
is $\Omega\left(d^{1/3}\right).$ 
\end{theorem}
\begin{proof}
Let $r$ be the number of messages in the protocol $M$ on inputs of size $d$. We consider the class of parity functions of strings of length $c$. We apply \cref{cl:LerningToSelection} to this class, where the examples are taken with uniform distribution. The length of the inputs for the selection problem is  $d=2^c$, the number of  parity functions.
This results in  a 
$1$-differentially private local protocol for learning the class of parity functions under the uniform distribution with $r$ messages.
By \cref{fact:learning_parity_LDP} (from~\cite{KLNRS11}),
the number of messages in such protocol is $\Omega(2^{c/3})=\Omega(d^{1/3})$.
\end{proof}

\end{document}